\documentclass[a4paper,UKenglish]{lipics-v2018}

\nolinenumbers
\hideLIPIcs

\usepackage{amsmath,amsfonts,amssymb,amsthm}
\usepackage{graphicx}
\usepackage{calc}
\usepackage{cite}
\usepackage[mode=text]{siunitx} 
\usepackage{balance}
\usepackage{caption} 
\usepackage{subcaption} 
\usepackage{float} 
\usepackage{yhmath} 

\usepackage{url}
\makeatletter
\g@addto@macro{\UrlBreaks}{\UrlOrds}
\makeatother

\usepackage{thmtools}
\usepackage{thm-restate}
\usepackage{hyperref}
\usepackage{cleveref} 

\usepackage{chngcntr}
\usepackage{apptools}
\usepackage{mathtools}

\newcommand{\overarc}{\wideparen}
\newcommand{\dist}{\textbf{d}}
\newcommand{\volfn}{\textrm{vol}}

\newcommand{\ax}{\ensuremath{g}}
\newcommand{\ux}{{\ensuremath{g^\uparrow}}}
\newcommand{\uxi}[1]{{\ensuremath{g^\uparrow_{#1}}}}
\newcommand{\lx}{{\ensuremath{g^\downarrow}}}
\newcommand{\lxi}[1]{{\ensuremath{g^\downarrow_{#1}}}}

\newcommand{\bxi}[1]{{\ensuremath{g^\updownarrow_{#1}}}}
\newcommand{\as}{\ensuremath{s}}

\newcommand{\usi}[1]{{\ensuremath{s^\uparrow_{#1}}}}
\newcommand{\ls}{{\ensuremath{s^\downarrow}}}
\newcommand{\lsi}[1]{{\ensuremath{s^\downarrow_{#1}}}}

\newcommand{\bsi}[1]{{\ensuremath{s^\updownarrow_{#1}}}}

\newcommand{\aseed}{\ensuremath{s}}

\newcommand{\seedset}{\ensuremath{\mathcal{S}}}

\newcommand{\lfs}{\textrm{lfs}}

\newcommand{\R}{\ensuremath{\mathbb{R}^3}}
\newcommand{\Rd}{\ensuremath{\mathbb{R}^d}}
\newcommand{\pd}[2]{\ensuremath{\pi(#1, #2)}}

\newcommand{\ddarrow}{\ensuremath{{\downarrow\downarrow}}}

\newcommand{\surf}{\ensuremath{\mathcal{M}}}
\newcommand{\newsurf}{\ensuremath{\hat{\mathcal{M}}}}

\newcommand{\vol}{\ensuremath{\mathcal{O}}} 
\newcommand{\newvol}{\ensuremath{\hat{\mathcal{O}}}} 
\newcommand{\sampleset}{\mathcal{P}}
\newcommand{\samp}{\ensuremath{p}}
\newcommand{\sampi}{\ensuremath{p_i}}

\newcommand{\ballunion}{\ensuremath{\mathcal{U}}}
\newcommand{\ballset}{\ensuremath{\mathcal{B}}}
\newcommand{\patchset}{\ensuremath{\mathcal{F}}}
\newcommand{\ball}{\ensuremath{B}}

\newcommand{\sphere}{\ensuremath{\partial\ball}}
\newcommand{\cir}[1]{\ensuremath{C_{#1}}}
\newcommand{\Vor}[1]{\ensuremath{\textrm{Vor}(#1)}}
\newcommand{\Vorz}{\ensuremath{\textrm{Vor}}}
\newcommand{\powz}{\ensuremath{\textrm{wVor}}}
\newcommand{\Delz}{\ensuremath{\textrm{Del}}}
\newcommand{\wdtz}{\ensuremath{\textrm{wDel}}}
\newcommand{\AlphaK}{\ensuremath{\mathcal{W}}}
\newcommand{\AlphaS}{\ensuremath{\mathcal{J}}}
\newcommand{\ucap}{\ensuremath{K^\uparrow}} 
\newcommand{\lcap}{\ensuremath{K^\downarrow}} 
\newcommand{\urim}{\ensuremath{\partial \ucap}} 
\newcommand{\lrim}{\ensuremath{\partial \lcap}} 

\newcommand{\G}{\ensuremath{\mathcal{S}}} 
\newcommand{\uGi}[1]{\ensuremath{\mathcal{S}^\uparrow_{#1}}} %
\newcommand{\lGi}[1]{\ensuremath{\mathcal{S}^\downarrow_{#1}}} 

\newtheorem{myLemma}[theorem]{Lemma}
\newtheorem{myCorollary}[theorem]{Corollary}
\newtheorem{myObservation}[theorem]{Observation}
\newtheorem{myClaim}[theorem]{Claim}

\AtAppendix{\counterwithin{apxLemma}{section}}
\newtheorem{apxLemma}{Lemma}

\AtAppendix{\counterwithin{apxCorollary}{section}}
\newtheorem{apxCorollary}[apxLemma]{Corollary}

\AtAppendix{\counterwithin{apxClaim}{section}}
\newtheorem{apxClaim}[apxLemma]{Claim}

\def\voroplusplus{{Voro\nolinebreak[4]\hspace{-.05em}\raisebox{.4ex}{\tiny\bf ++}}}

\newcommand{\mytitle}{Sampling Conditions for Conforming Voronoi Meshing by the VoroCrust Algorithm}
\title{\mytitle}
\titlerunning{\mytitle}

\author{Ahmed Abdelkader}{University of Maryland, College Park MD, USA}{akader@cs.umd.edu}{https://orcid.org/0000-0002-6749-1807}{}

\author{Chandrajit L. Bajaj}{University of Texas, Austin TX, USA}{}{}{Supported in part by a contract from Sandia, \#1439100, and a grant from NIH, R01-GM117594}

\author{Mohamed S. Ebeida}{Sandia National Laboratories, Albuquerque NM, USA}{}{}{}

\author{Ahmed H. Mahmoud}{University of California, Davis CA, USA}{}{}{}

\author{Scott A. Mitchell}{Sandia National Laboratories, Albuquerque NM, USA}{}{}{}

\author{John D. Owens}{University of California, Davis CA, USA}{}{}{}

\author{Ahmad A. Rushdi}{University of California, Davis CA, USA}{}{}{}

\authorrunning{A. Abdelkader et al.}

\Copyright{Ahmed Abdelkader, Chandrajit L. Bajaj, Mohamed S. Ebeida, Ahmed H. Mahmoud, \\ Scott A. Mitchell, John D. Owens and Ahmad A. Rushdi}

\subjclass{Theory of computation $\rightarrow$ Randomness, geometry and discrete structures $\rightarrow$ Computational geometry}


\category{}

\relatedversion{This is an extended version of the original paper which appeared in the Symposium on Computational Geometry~\cite{VC_SoCG18}, and includes appendices with proofs that the original paper did not. See also the accompanying multimedia contribution~\cite{SoCG_MME}.}

\supplement{}

\funding{
This material is based upon work supported by the U.S. Department of Energy, Office of Science, Office of Advanced Scientific Computing Research (ASCR), Applied Mathematics Program. Sandia National Laboratories is a multi-mission laboratory managed and operated by National Technology and Engineering Solutions of Sandia, LLC., a wholly owned subsidiary of Honeywell International, Inc., for the U.S. Department of Energy's National Nuclear Security Administration under contract DE-NA0003525. The views expressed in the article do not necessarily represent the views of the U.S. Department of Energy or the United States Government.}

\acknowledgements{We thank Tamal Dey for helpful discussions about surface reconstruction.}

\EventEditors{Bettina Speckmann and Csaba D. T{\'o}th}
\EventNoEds{2}
\EventLongTitle{34th International Symposium on Computational Geometry (SoCG 2018)}
\EventShortTitle{SoCG 2018}
\EventAcronym{SoCG}
\EventYear{2018}
\EventDate{June 11--14, 2018}
\EventLocation{Budapest, Hungary}
\EventLogo{socg-logo} 
\SeriesVolume{99}
\ArticleNo{1}

\begin{document}

\maketitle


\begin{abstract}
We study the problem of decomposing a volume bounded by a smooth surface into a collection of Voronoi cells. Unlike the dual problem of conforming Delaunay meshing, a principled solution to this problem for generic smooth surfaces remained elusive. VoroCrust leverages ideas from $\alpha$-shapes and the power crust algorithm to produce unweighted Voronoi cells conforming to the surface, yielding the first provably-correct algorithm for this problem. Given an $\epsilon$-sample on the bounding surface, with a weak $\sigma$-sparsity condition, we work with the balls of radius $\delta$ times the local feature size centered at each sample. The corners of this union of balls are the Voronoi sites, on both sides of the surface. The facets common to cells on opposite sides reconstruct the surface. For appropriate values of $\epsilon$, $\sigma$ and $\delta$, we prove that the surface reconstruction is isotopic to the bounding surface. With the surface protected, the enclosed volume can be further decomposed into an isotopic volume mesh of fat Voronoi cells by generating a bounded number of sites in its interior. Compared to state-of-the-art methods based on clipping, VoroCrust cells are full Voronoi cells, with convexity and fatness guarantees. Compared to the power crust algorithm, VoroCrust cells are not filtered, are unweighted, and offer greater flexibility in meshing the enclosed volume by either structured grids or random samples.
\end{abstract}

\section{Introduction} \label{sec:intro}
Mesh generation is a fundamental problem in computational geometry, geometric modeling, computer graphics, scientific computing and engineering simulations. There has been a growing interest in polyhedral meshes as an alternative to tetrahedral or hex-dominant meshes~\cite{cdadapco:polyhedral}.  Polyhedra are less sensitive to stretching, which enables the representation of complex geometries without excessive refinement. In addition, polyhedral cells have more neighbors even at corners and boundaries, which offers better approximations of gradients and local flow distributions. Even compared to hexahedra, fewer polyhedral cells are needed to achieve a desired accuracy in certain applications. This can be very useful in several numerical methods~\cite{doi:10.1142/S0218202514030018}, e.g., finite element~\cite{doi:10.1142/S0218202514400065}, finite volume~\cite{FVM}, virtual element~\cite{VEM} and Petrov-Galerkin~\cite{kuzmin2010guide}. In particular, the accuracy of a number of important solvers, e.g., the two-point flux approximation for conservation laws~\cite{FVM}, greatly benefits from a conforming mesh which is \textit{orthogonal} to its dual as naturally satisfied by Voronoi meshes. Such solvers play a crucial role in hydrology~\cite{Sents:2017}, computational fluid dynamics~\cite{Brochu:2010} and fracture modeling~\cite{Bishop:2009:STP}.

VoroCrust is the first provably-correct algorithm for generating a volumetric Voronoi mesh whose boundary conforms to a smooth bounding surface, and with quality guarantees.
A conforming volume mesh exhibits two desirable properties \emph{simultaneously}: (1) a decomposition of the enclosed volume, and (2) a reconstruction of the bounding surface.

Conforming Delaunay meshing is well-studied~\cite{cheng2012delaunay}, but Voronoi meshing is less mature. A common practical approach to polyhedral meshing is to dualize a tetrahedral mesh and \emph{clip}, i.e., intersect and truncate, each cell by the bounding surface \cite{okabe2009spatial,Ebeida:2011:URVM,Merland2014,Si2010,yan2013efficient}. Unfortunately, clipping sacrifices the important properties of convexity and connectedness of cells~\cite{Ebeida:2011:URVM}, and may require costly constructive solid geometry operations. Restricting a Voronoi mesh to the surface before \textit{filtering} its dual Delaunay facets is another approach~\cite{amenta1999surface,dey2013voronoi,Yan:2009:IRF:1735603.1735629}, but filtering requires extra checks complicating its implementation and analysis; see also \Cref{fig:sandwiching}. An intuitive approach is to locally mirror the Voronoi sites on either side of the surface~\cite{mirroring_Bolander,7298931}, but we are not aware of any robust algorithms with approximation guarantees in this category. In contrast to these approaches, VoroCrust is distinguished by its simplicity and robustness at producing true unweighted Voronoi cells, leveraging established libraries, e.g., \voroplusplus~\cite{rycroft2009voro++}, without modification or special cases.

VoroCrust can be viewed as a principled mirroring technique, which shares a number of key features with the power crust algorithm~\cite{Amenta2001127}. The power crust literature~\cite{amenta1998crust,amenta1999surface,Amenta:1998:NVS:280814.280947,Amenta:2001:PC:376957.376986,Amenta2001127} developed a rich theory for surface approximation, namely the $\epsilon$-sampling paradigm. Recall that the power crust algorithm uses an $\epsilon$-sample of unweighted points to place weighted sites, so-called \textit{poles}, near the medial axis of the underlying surface. The surface reconstruction is the collection of facets separating power cells of poles on the inside and outside of the enclosed volume.

Regarding samples and poles as primal-dual constructs, power crust performs a \emph{primal-dual-dual-primal dance}. VoroCrust makes a similar dance where weights are introduced differently; the samples are weighted to define unweighted sites tightly hugging the surface, with the reconstruction arising from their unweighted Voronoi diagram. The key advantage is the freedom to place more sites within the enclosed volume without disrupting the surface reconstruction. This added freedom is essential to the generation of graded meshes; a primary virtue of the proposed algorithm. Another virtue of the algorithm is that all samples appear as vertices in the resulting mesh. While the power crust algorithm does not guarantee that, some variations do so by means of filtering, at the price of the reconstruction no longer being the boundary of power cells~\cite{amenta1999surface,Amenta:2000:SAH:336154.336207,dey2009isotopic}.

The main construction underlying VoroCrust is a suitable union of balls centered on the bounding surface, as studied in the context of non-uniform approximations~\cite{chazal2006topology}. Unions of balls enjoy a wealth of results~\cite{Edelsbrunner:1993:UBD:160985.161139,AMENTA200125,Cazals:2011:CVU:2049662.2049665}, which enable a variety of algorithms~\cite{Amenta2001127,CGF:CGF12270,LetscherMA}.

Similar constructions have been proposed for meshing problems in the applied sciences with heuristic extensions to 3D settings; see~\cite{Klemetsdal2017} and the references therein for a recent example. Aichholzer et al.~\cite{CGF:CGF1512} adopt closely related ideas to construct a union of surface balls using power crust poles for sizing estimation. However, their goal was to produce a coarse homeomorphic surface reconstruction. As in~\cite{CGF:CGF1512}, the use of balls and $\alpha$-shapes for surface reconstruction was explored earlier, e.g., ball-pivoting~\cite{817351,stelldinger2008topologically}, but the connection to Voronoi meshing has been absent. In contrast, VoroCrust aims at a decomposition of the enclosed volume into fat Voronoi cells conforming to an isotopic surface reconstruction with quality guarantees.

In a previous paper~\cite{shatter}, we explored the related problem of generating a Voronoi mesh that conforms to restricted classes of piecewise-linear complexes, with more challenging inputs left for future work. The approach adopted in~\cite{shatter} does not use a union of balls and relies instead on similar ideas to those proposed for conforming Delaunay meshing~\cite{MMG2000,Cohen-Steiner:2002,Rand2009}.

In this paper, we present a theoretical analysis of an abstract version of the VoroCrust algorithm. This establishes the quality and approximation guarantees of its output for volumes bounded by smooth surfaces. A description of the algorithm we analyze is given next; see \Cref{fig:intro_bd} for an illustration in 2D, and also our accompanying multimedia contribution~\cite{SoCG_MME}.

\subsection*{The abstract VoroCrust algorithm}\label{sec:reconstruction_steps}
\begin{enumerate}
\item Take as input a sample $\sampleset$ on the surface $\surf$ bounding the volume $\vol$.
\label{step:initsample}
\item Define a ball $\ball_i$ centered at each sample $\samp_i$, with a suitable radius $r_i$, and let $\ballunion = \cup_i \ball_i$.
\item Initialize the set of sites $\seedset$ with the corner points of $\partial \ballunion$, $\seedset^\uparrow$ and $\seedset^\downarrow$, on both sides of $\surf$.
\item Optionally, generate additional sites $\seedset^\ddarrow$ in the interior of $\vol$, and include $\seedset^\ddarrow$ into $\seedset$.
\item Compute the Voronoi diagram $\Vorz(\seedset)$ and retain the cells with sites in $\seedset^\downarrow \cup \seedset^\ddarrow$ as the volume mesh $\newvol$, where the facets between $\seedset^\uparrow$ and $\seedset^\downarrow$ yield a surface approximation $\newsurf$.
\end{enumerate}

\begin{figure}[H]
\centering
\begin{subfigure}[b]{0.24\columnwidth}\centering
  \includegraphics[width=0.95\hsize]{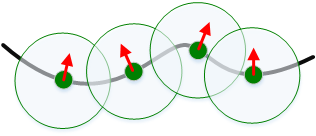}
  \caption{Surface balls.}
\end{subfigure}
\begin{subfigure}[b]{0.24\columnwidth}\centering
  \includegraphics[width=0.95\hsize]{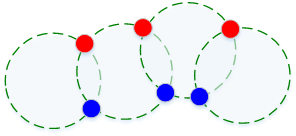}
  \caption{Labeled corners.}
\end{subfigure}
\begin{subfigure}[b]{0.24\columnwidth}\centering
  \includegraphics[width=0.95\hsize]{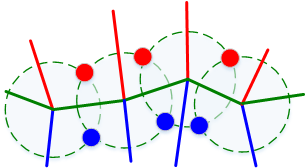}
  \caption{Voronoi cells.}
\end{subfigure}
\begin{subfigure}[b]{0.24\columnwidth}\centering
  \includegraphics[width=0.95\hsize]{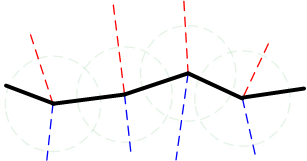}
  \caption{Reconstruction.}
\end{subfigure}
   \caption{VoroCrust reconstruction, demonstrated on a planar curve.}
   \label{fig:intro_bd}
\end{figure}

In this paper, we assume $\vol$ is a bounded open subset of $\R$, whose boundary $\surf$ is a closed, bounded and smooth surface. We further assume that $\sampleset$ is an $\epsilon$-sample, with a weak $\sigma$-sparsity condition, and $r_i$ is set to $\delta$ times the local feature size at $p_i$. For appropriate values of $\epsilon$, $\sigma$ and $\delta$, we prove that $\newvol$ and $\newsurf$ are isotopic to $\vol$ and $\surf$, respectively. We also show that simple techniques for sampling within $\vol$, e.g., octree refinement, guarantee an upper bound on the fatness of all cells in $\newvol$, as well as the number of samples.

Ultimately, we seek a conforming Voronoi mesher that can handle realistic inputs possibly containing sharp features, can estimate a sizing function and generate samples, and can guarantee the quality of the output mesh. This is the subject of a forthcoming paper~\cite{vorocrust_02_surface_sampling} which describes the design and implementation of the complete VoroCrust algorithm.

The rest of the paper is organized as follows. Section 2 introduces the key definitions and notation used throughout the paper. Section 3 describes the placement of Voronoi seeds and basic properties of our construction assuming the union of surface balls satisfies a structural property. Section 4 proves this property holds and establishes the desired approximation guarantees under certain conditions on the input sample. Section 5 considers the generation of interior samples and bounds the fatness of all cells in the output mesh. Section 6 concludes the paper with pointers for future work. A number of proofs are deferred to the appendices; see also the accompanying multimedia contribution~\cite{SoCG_MME}.

\section{Preliminaries} \label{sec:prelim}

Throughout, standard general position assumptions~\cite{SoS} are made implicitly to simplify the presentation. We use $\dist(p, q)$ to denote the Euclidean distance between two points $p, q \in \R$, and $\mathbb{B}(c, r)$ to denote the Euclidean ball centered at $c \in \R$ with radius $r$. We proceed to introduce the notation and recall the key definitions used throughout, following those in~\cite{Edelsbrunner:1993:UBD:160985.161139,Amenta2001127,chazal2006topology}.

\subsection{Sampling and approximation}
We take as input a set of sample points $\sampleset \subset \surf$. A local scale or \textit{sizing} is used to vary the sample density. Recall that the \textit{medial axis}~\cite{Amenta2001127} of $\surf$, denoted by $\mathcal{A}$, is the closure of the set of points in $\R$ with more than one closest point on $\surf$. Hence, $\mathcal{A}$ has one component inside $\vol$ and another outside.  Each point of $\mathcal{A}$ is the center of a \textit{medial ball} tangent to $\surf$ at multiple points. Likewise, each point on $\surf$ has two tangent medial balls, not necessarily of the same size. The \textit{local feature size} at $x \in \surf$ is defined as $\lfs(x) = \inf _{{a\in \mathcal{A}}} \dist(x, a)$. The set $\sampleset$ is an \textit{$\epsilon$-sample}~\cite{amenta1999optimal} if for all $x \in \surf$ there exists $p \in \sampleset$ such that $\dist(x, p) \leq \epsilon\cdot\lfs(x)$.

We desire an approximation of $\vol$ by a Voronoi mesh $\newvol$, where the boundary $\newsurf$ of $\newvol$ approximates $\surf$. Recall that two topological spaces are \textit{homotopy-equivalent}~\cite{chazal2006topology} if they have the same topology type.
A stronger notion of topological equivalence is \textit{homeomorphism}, which holds when there exists a continuous bijection with a continuous inverse from $\surf$ to $\newsurf$.
The notion of isotopy captures an even stronger type of equivalence for surfaces \textit{embedded} in Euclidean space. Two surfaces $\surf, \newsurf \subset \R$ are \textit{isotopic}~\cite{AMENTA20033,CHAZAL2005390} if there is a continuous mapping $F: \surf \times [0, 1] \to \R$ such that for each $t \in [0, 1]$, $F(\cdot, t)$ is a homeomorphism from $\surf$ to $\newsurf$, where $F(\cdot, 0)$ is the identity of $\surf$ and $F(\surf, 1) = \newsurf$. To establish that two surfaces are \textit{geometrically close}, the distance between each point on one surface and its closest point on the other surface is required. Such a bound is usually obtained in the course of proving isotopy.

\subsection{Diagrams and triangulations}
The set of points defining a Voronoi diagram are traditionally referred to as \textit{sites} or \textit{seeds}. When approximating a manifold by a set of sample points of varying density, it is helpful to assign weights to the points reflective of their density. In particular, a point $\samp_i$ with weight $w_i$, can be regarded as a ball $\ball_i$ with center $p_i$ and radius $r_i = \sqrt{w_i}$.

Recall that the \textit{power distance}~\cite{Edelsbrunner:1993:UBD:160985.161139} between two points $\samp_i, \samp_j$ with weights $w_i, w_j$ is $\pd{\samp_i}{\samp_j}= \dist(\samp_i, \samp_j)^2 - w_i - w_j$. Unless otherwise noted, points are \textit{unweighted}, having weight equal to zero. There is a natural geometric interpretation of the weight: all points $q$ on the boundary of $\ball_i$ have $\pd{\sampi}{q}=0$, inside $\pd{\sampi}{q}<0$ and outside $\pd{\sampi}{q}>0.$ Given a set of weighted points $\sampleset$, this metric gives rise to a natural decomposition of $\R$ into the \textit{power cells} $V_i = \{q \in \R \mid \pd{\samp_i}{q} \leq \pd{\samp_j}{q} \: \forall \samp_j \in \sampleset \}$. The \textit{power diagram} $\powz(\sampleset)$ is the cell complex defined by collection of cells $V_i$ for all $p_i \in \sampleset$.

The nerve~\cite{Edelsbrunner:1993:UBD:160985.161139} of a collection $\mathcal{C}$ of sets is defined as $\mathcal{N}(\mathcal{C}) = \{X \subseteq \mathcal{C} \mid \cap\:T \neq \emptyset \}$. Observe that $\mathcal{N}(\mathcal{C})$ is an abstract simplicial complex because $X \in \mathcal{N}(\mathcal{C})$ and $Y \subseteq X$ imply $Y \in \mathcal{N}(\mathcal{C})$. With that, we obtain the \textit{weighted Delaunay triangulation}, or \textit{regular triangulation}, as $\wdtz(\sampleset) = \mathcal{N}(\powz(\sampleset))$. Alternatively, $\wdtz(\sampleset)$ can be defined directly as follows. A subset $T \subset \Rd$, with $d \leq 3$ and $|T| \leq d+1$ defines a $d$-simplex $\sigma_T$. Recall that the \textit{orthocenter}~\cite{exudation} of $\sigma_T$, denoted by $z_T$, is the unique point $q \in \Rd$ such that $\pd{\samp_i}{z_T} = \pd{\samp_j}{z_T}$ for all $\samp_i, \samp_j \in T$; the \textit{orthoradius} of $\sigma_T$ is equal to $\pd{\samp}{z_T}$ for any $\samp \in T$.  The \textit{Delaunay condition} defines $\wdtz(\sampleset)$ as the set of tetrahedra  $\sigma_T$ with an \textit{empty orthosphere}, meaning $\pd{\samp_i}{z_T} \leq \pd{\samp_j}{z_T}$ for all $\samp_i \in T$ and $\samp_j \in \sampleset \setminus T$, where $\wdtz(\sampleset)$ includes all faces of $\sigma_T$.

There is a natural duality between $\wdtz(\sampleset)$ and $\powz(\sampleset)$. For a tetrahedron $\sigma_T$, the definition of $z_T$ immediately implies $z_T$ is a \textit{power vertex} in $\powz(\sampleset)$. Similarly, for each $k$-face $\sigma_S$ of $\sigma_T \in \wdtz(\sampleset)$ with $S \subseteq T$ and $k + 1 = |S|$, there exists a dual $(3-k)$-face $\sigma'_S$ in $\powz(\sampleset)$ realized as $\cap_{\samp \in S} V_\samp$. When $\sampleset$ is unweighted, the same definitions yield the standard (unweighted) Voronoi diagram $\Vorz(\sampleset)$ and its dual Delaunay triangulation $\Delz(\sampleset)$.

\subsection{Unions of balls}
Let $\ballset$ denote the set of balls corresponding to a set of weighted points $\sampleset$ and define the \textit{union of balls} $\ballunion$ as $\cup \ballset$. It is quite useful to capture the structure of $\ballunion$ using a combinatorial representation like a simplicial complex~\cite{Edelsbrunner:1993:UBD:160985.161139,edelsbrunner1992weighted}. Let $f_i$ denote $V_i \cap \partial\ball_i$ and $\patchset$ the collection of all such $f_i$. Observing that $V_i \cap \ball_j \subseteq V_i \cap \ball_i\:\forall \ball_i, \ball_j \in \ballset$, $f_i$ is equivalently defined as the spherical part of $\partial(V_i \cap \ball_i)$. Consider also the decomposition of $\ballunion$ by the cells of $\powz(\sampleset)$ into $\mathcal{C}(\ballset) = \{ V_i \cap \ball_i \mid \ball_i \in \ballset \}$. The \textit{weighted $\alpha$-complex} $\AlphaK(\sampleset)$ is defined as the \textit{geometric realization} of $\mathcal{N}(\mathcal{C}(\ballset))$~\cite{Edelsbrunner:1993:UBD:160985.161139}, i.e., $\sigma_T \in \AlphaK$ if $\{V_i \cap \ball_i \mid \samp_i \in T\} \in \mathcal{N}(\mathcal{C}(\ballset))$. It is not hard to see that $\AlphaK$ is a subcomplex of $\wdtz(\sampleset)$.

To see why $\AlphaK$ is relevant, consider its \textit{underlying space}; we create a collection containing the convex hull of each simplex in $\AlphaK$ and define the \textit{weighted $\alpha$-shape} $\AlphaS(\sampleset)$ as the union of this collection. It turns out that the simplices $\sigma_T \in \AlphaK$ contained in $\partial\AlphaS$ are dual to the faces of $\partial\ballunion$ defined as $\cap_{i \in T}f_i$. Every point $q \in \partial\ballunion$ defined by $\cap_{i \in T_q}f_i$, for $T_q \in \ballset$ and $k + 1 = |T_q|$, witnesses the existence of $\sigma_{T_q}$ in $\AlphaK$; the $k$-simplex $\sigma_{T_q}$ is said to be \textit{exposed} and $\partial\AlphaS$ can be defined directly as the collection of all exposed simplices~\cite{edelsbrunner1992weighted}. In particular, the \textit{corners} of $\partial\ballunion$ correspond to the facets of $\partial\AlphaS$. Moreover, $\AlphaS$ is homotopy-equivalent to $\ballunion$~\cite{Edelsbrunner:1993:UBD:160985.161139}.

The union of balls defined using an $\epsilon$-sampling guarantees the approximation of the manifold under suitable conditions on the sampling. Following earlier results on uniform sampling~\cite{Niyogi2008}, an extension to non-uniform sampling establishes sampling conditions for the isotopic approximation of hypersurfaces and medial axis reconstruction~\cite{chazal2006topology}.

\section{Seed placement and surface reconstruction}\label{sec:sandwiching}
We determine the location of Voronoi seeds using the union of balls $\ballunion$. The correctness of our reconstruction depends crucially on how sample balls $\ballset$ overlap. Assuming a certain structural property on $\ballunion$, the surface reconstruction is embedded in the dual shape $\AlphaS$.

\subsection{Seeds and guides}\label{sec:guides}
Central to the method and analysis are triplets of sample spheres, i.e., boundaries of sample balls, corresponding to a \emph{guide triangle} in $\wdtz(\sampleset)$. The sample spheres associated with the vertices of a guide triangle intersect contributing a pair of \emph{guide points}. The reconstruction consists of Voronoi facets, most of which are guide triangles.

\begin{figure}[!htb]
\centering
\begin{subfigure}[b]{0.44\columnwidth}\centering
  \includegraphics[width=0.8\hsize]{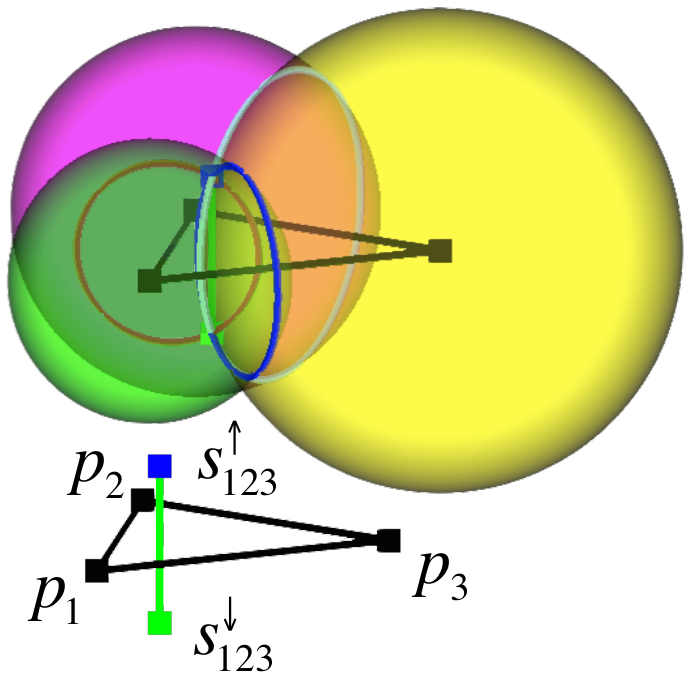}
  \caption{Overlapping balls and guide circles.}
  \label{fig:triplet}
\end{subfigure}
\hfill
\begin{subfigure}[b]{0.55\columnwidth}\centering
  \includegraphics[width=0.95\hsize]{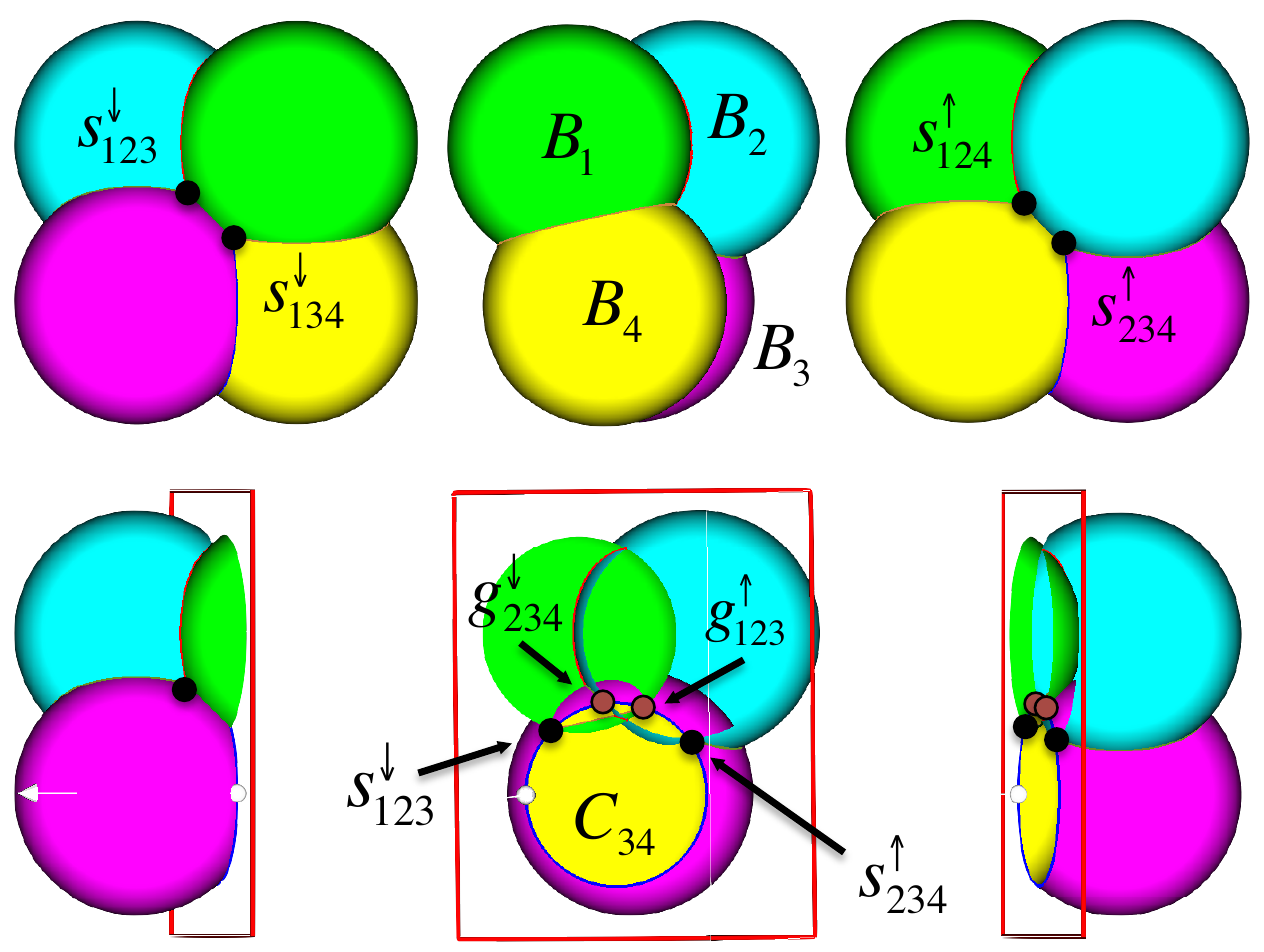}
  \caption{Pattern resulting in four half-covered seed pairs.}
    \label{fig:sliverspheres}
\end{subfigure}
   \caption{(a) Guide triangle and its dual seed pair. (b) Cutaway view in the plane of circle $\cir{34}$.}
   \label{fig:basics}
\end{figure}

When a triplet of spheres $\sphere_i, \sphere_j, \sphere_k$ intersect at exactly two points, the intersection points are denoted by $\bxi{ijk} = \{\uxi{ijk},\lxi{ijk}\}$ and called a pair of \emph{guide points} or \emph{guides}; see  \Cref{fig:triplet}.
The associated \textit{guide triangle} $t_{ijk}$ is \emph{dual} to $\bxi{ijk}$. We use arrows to distinguish guides on different sides of the manifold with the \textit{upper} guide $\ux$ lying outside $\vol$ and the \textit{lower} guide $\lx$ lying inside. We refer to the edges of guide triangles as \emph{guide edges} $e_{ij} = \overline{\samp_i\samp_j}$. A guide edge $e_{ij}$ is associated with a dual \emph{guide circle} $\cir{ij} = \sphere_i \cap \sphere_j$, as in~\Cref{fig:triplet}.

The Voronoi seeds in $\seedset^\uparrow \cup \seedset^\downarrow$ are chosen as the subset of guide points that lie on $\partial\ballunion$. A guide point $\ax$ which is not interior to any sample ball is \emph{uncovered} and included as a \emph{seed} $\as$ into $\seedset$; covered guides are not.
We denote \emph{uncovered guides} by $s$ and \emph{covered guides} by $g$, whenever coverage is known and important.
If only one guide point in a pair is covered, then we say the guide pair is \emph{half-covered}.
If both guides in a pair are covered, they are ignored. Let $\seedset_i = \seedset \cap \sphere_i$ denote the seeds on sample sphere $\sphere_i$.

As each guide triangle $t_{ijk}$ is associated with at least one dual seed $s_{ijk}$, the seed witnesses its inclusion in $\AlphaK$ and $t_{ijk}$ is exposed. Hence, $t_{ijk}$ belongs to $\partial\AlphaS$ as well. When such $t_{ijk}$ is dual to a single seeds $s_{ijk}$ it bounds the interior of $\AlphaS$, i.e., it is a face of a \textit{regular component} of $\AlphaS$; in the simplest and most common case, $t_{ijk}$ is a facet of a tetrahedron as shown in \Cref{fig:sliverfacets}. When $t_{ijk}$ is dual to a pair of seeds $\bsi{ijk}$, it does not bound the interior of $\AlphaS$ and is called a \textit{singular face} of $\partial\AlphaS$. All singular faces of $\partial\AlphaS$ appear in the reconstructed surface.

\subsection{Disk caps}\label{sec:disk_caps}
We describe the structural property required on $\ballunion$ along with the consequences exploited by VoroCrust for surface reconstruction. This is partially motivated by the requirement that all sample points on the surface appear as vertices in the output Voronoi mesh.

We define the subset of $\sphere_i$ inside other balls as the \emph{medial band} and say it is \emph{covered}.
Let the caps $\ucap_i$ and $\lcap_i$ be the complement of the medial band in the interior and exterior of $\vol$, respectively. 
Letting $n_{p_i}$ be the normal line through $p_i$ perpendicular to $\surf$, the two intersection points $n_{p_i} \cap \sphere_i$ are called the \emph{poles} of $\ball_i$.
See \Cref{fig:capscrowns}.

\begin{figure}[!htb]
\centering
\begin{subfigure}[b]{0.38\columnwidth}\centering
  \includegraphics[width=0.97\hsize]{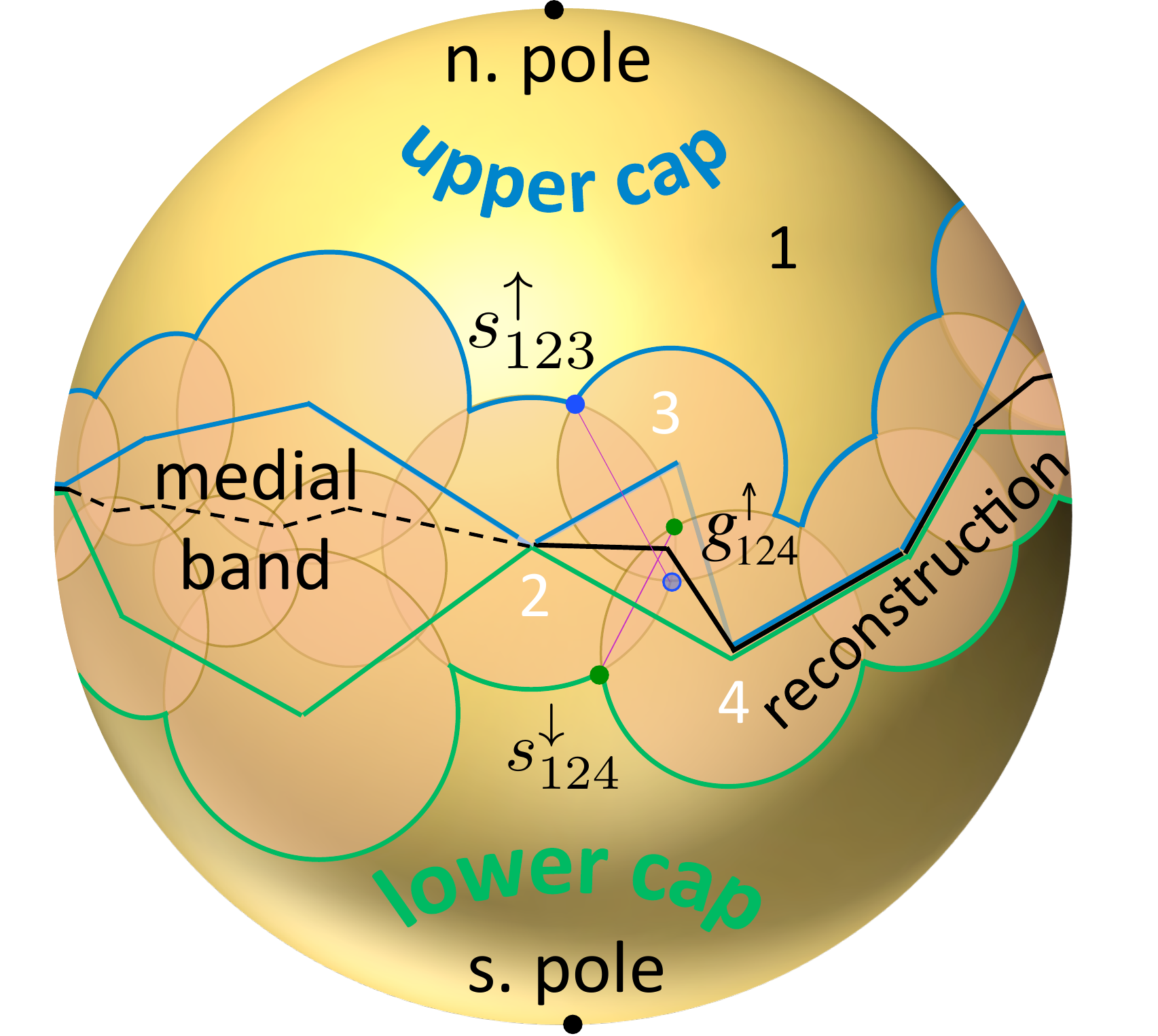}
  \caption{Caps and medial band.}
  \label{fig:capscrowns}
\end{subfigure}
\hfill
\begin{subfigure}[b]{0.61\columnwidth}\centering
  \includegraphics[width=0.999\hsize]{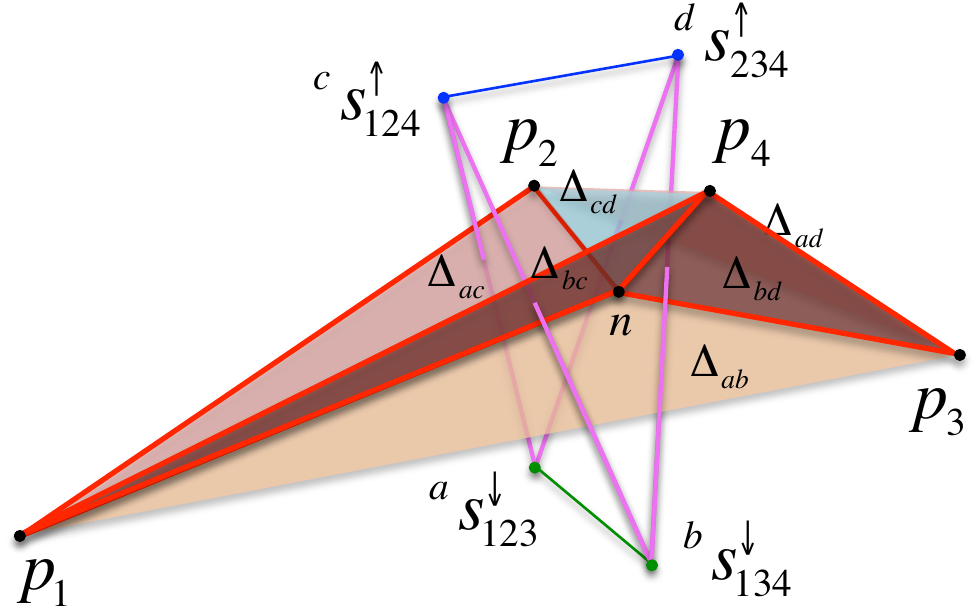}
  \caption{Sliver and half-covered seeds, exaggerated vertical scale.}
  \label{fig:sliverfacets}
\end{subfigure}
   \caption{(a) Decomposing the sample sphere $\partial\ball_1$. (b) Uncovered seeds and reconstruction facets. Let $\tau_p \in \AlphaK(\sampleset) \subseteq \wdtz(\sampleset)$ and $\tau_s \in \Delz(\seedset)$ denote the tetrahedra connecting the four samples and the four seeds shown, respectively. $\lsi{123}$ and $\lsi{134}$ are the uncovered lower guide seeds, with $\uxi{123}$ and $\uxi{134}$ covered. The uncovered upper guide seeds are $\usi{124}$ and $\usi{234}$, with  $\lxi{124}$ and $\lxi{234}$ covered. $\triangle_{ac}$ is the Voronoi facet dual to the Delaunay edge between $^a\lsi{123}$ and $^c\usi{124}$, etc. Voronoi facets dual to magenta edges are in the reconstructed surface; those dual to green and blue edges are not. $n$ is the circumcenter of $\tau_s$ and appears as a Voronoi vertex in $\Vorz(\seedset)$ and a \textit{Steiner vertex} in the surface reconstruction. In general, $n$ is not the orthocenter of the sliver $\tau_p$.}
   \label{fig:caps_sliver}
\end{figure}

We require that $\ballunion$ satisfies the following structural property: each $\sphere_i$ has \emph{disk caps}, meaning the medial band is a \textit{topological annulus} and the two caps contain the poles and are \textit{topological disks}. In other words, each $\ball_i$ contributes one connected component to each side of $\partial\ballunion$. As shown in~\Cref{fig:capscrowns}, all seeds in $\uGi{i}$ and $\lGi{i}$ lie on $\partial\ucap_i$ and $\partial\lcap_i$, respectively, along the arcs where other sample balls intersect $\sphere_i$. In \Cref{sec:density}, we establish sufficient sampling conditions to ensure $\ballunion$ satisfies this property. In particular, we will show that both poles of each $\ball_i$ lie on $\partial\ballunion$.

The importance of disk caps is made clear by the following observation. The requirement that all sample points appear as Voronoi vertices in $\newsurf$ follows as a corollary.

\begin{myObservation}[Three upper/lower seeds] \label{obs:threeupperseeds}
If $\sphere_i$ has \emph{disk caps}, then each of $\partial \ucap_i$ and $\partial \lcap_i$ has at least three seeds and the seeds on $\sphere_i$ are not all coplanar.
\end{myObservation}

\begin{proof}
Every sphere $S_{j \ne i}$ covers strictly less than one hemisphere of $\sphere_i$ because the poles are uncovered. 
Hence, each cap is composed of at least three arcs connecting at least three upper seeds $\uGi{i} \subset \partial \ucap_i$ and three lower seeds $\lGi{i} \subset \partial \lcap_i$.
Further, any hemisphere through the poles contains at least one upper and one lower seed.
It follows that the set of seeds $\seedset_i = \G_i^\uparrow \cup \G_i^\downarrow$ is not coplanar.
\end{proof}

\begin{myCorollary}[Sample reconstruction] \label{cor:sample_reconstruction}
If $\sphere_i$ has disk caps, then $\samp_i$ is a vertex in $\newsurf$.
\end{myCorollary}
\begin{proof}
By \Cref{obs:threeupperseeds}, the sample is equidistant to at least four seeds which are not all coplanar. It follows that the sample appears as a vertex in the Voronoi diagram and not in the relative interior of a facet or an edge. Being a common vertex to at least one interior and one exterior Voronoi seed, VoroCrust retains this vertex in its output reconstruction.
\end{proof}

\subsection{Sandwiching the reconstruction in the dual shape of \texorpdfstring{$\ballunion$}{\textit{U}}} \label{sec:sandwich}

Triangulations of smooth surfaces embedded in $\R$ can have half-covered guides pairs, with one guide covered by the ball of a fourth sample not in the guide triangle dual to the guide pair.
The tetrahedron formed by the three samples of the guide triangle plus the fourth covering sample is a \emph{sliver}, i.e., the four samples lie almost uniformly around the equator of a sphere.
In this case we do not reconstruct the guide triangle, and also do not reconstruct some guide edges. We show that the reconstructed surface $\newsurf$ lies entirely within the region of space bounded by guide triangles, i.e., the $\alpha$-shape of $\sampleset$, as stated in the following theorem.

\begin{restatable}{theorem}{sandwiching} \label{thm:sphere_sandwich}
If all sample balls have disk caps, then $\newsurf \subseteq \AlphaS(\sampleset)$.
\end{restatable}

\begin{figure}[H]
  \includegraphics[width=1.0\columnwidth]{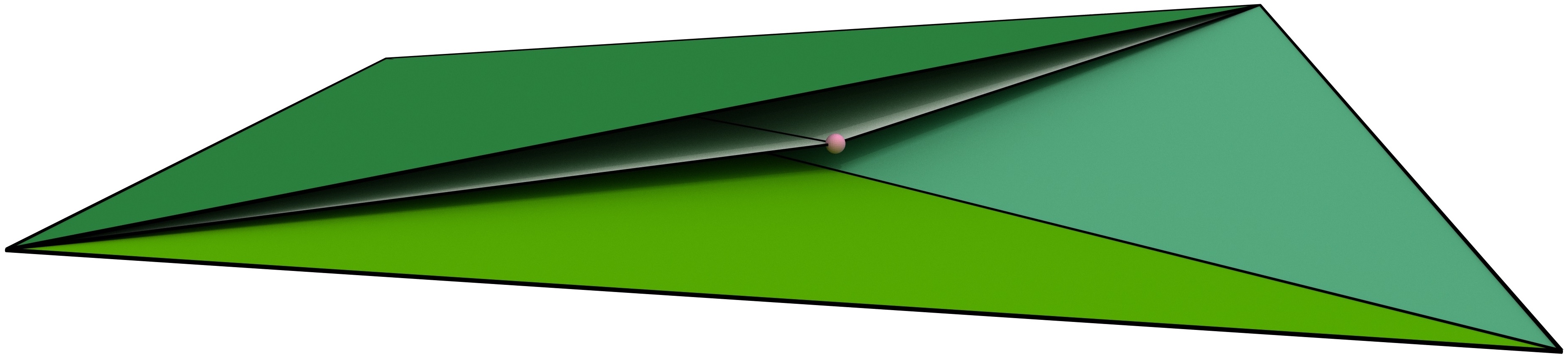}
  \caption{Cutaway view of a sliver tetrahedron $\tau_p \in \AlphaK(\sampleset) \subseteq \wdtz(\sampleset)$, drawn to scale. Half-covered guides give rise to the Steiner vertex (pink), which results in a surface reconstruction using four facets (only two are shown) sandwiched within $\tau_p$. In contrast, filtering $\wdtz(\sampleset)$ chooses two of the four facets of $\tau_p$, either the bottom two, or the top two (only one is shown).}
 \label{fig:sandwiching}
\end{figure}

The simple case of a single isolated sliver tetrahedron is illustrated in Figures~\ref{fig:sliverfacets},~\ref{fig:sandwiching} and~\ref{fig:sliverspheres}.
A sliver has a pair of lower guide triangles and a pair of upper guide triangles. 
For instance, $t_{124}$ and $t_{234}$ are the pair of upper triangles in \Cref{fig:sliverfacets}.
In such a tetrahedron, there is an edge between each pair of samples corresponding to a non-empty circle of intersection between sample balls, like the circles in \Cref{fig:triplet}.
For this circle, the arcs covered by the two other sample balls of the sliver overlap, so each of these balls contributes exactly one uncovered seed, rather than two.
In this way the upper guides for the upper triangles are uncovered, but their lower guides are covered; also only the lower guides of the lower triangles are uncovered.
The proof of \Cref{thm:sphere_sandwich} follows by analyzing the Voronoi cells of the seed points located on the overlapping sample balls and is deferred to~\Cref{sec:sandwich_analysis}. 
Alternatively, \Cref{thm:sphere_sandwich} can be seen as a consequence of Theorem~2 in~\cite{AMENTA200125}.

\section{Sampling conditions and approximation guarantees} \label{sec:sampling} \label{sec:density}

We take as input a set of points $\sampleset$ sampled from the bounding surface $\surf$ such that $\sampleset$ is an $\epsilon$-sample, with $\epsilon \leq 1/500$. We require that $\sampleset$ satisfies the following sparsity condition: for any two points $p_i, p_j \in P$, $\lfs(p_i) \geq \lfs(p_j) \implies \dist(p_i, p_j) \geq \sigma \epsilon \lfs(p_j)$, with $\sigma \geq 3/4$.

Such a sampling $\sampleset$ can be obtained by known algorithms. Given a suitable representation of $\surf$, the algorithm in~\cite{BOISSONNAT2005405} computes a loose $\epsilon'$-sample $E$ which is a $\epsilon'(1+8.5\epsilon')$-sample. More specifically, whenever the algorithm inserts a new sample $p$ into the set $E$, $\dist(p, E) \geq \epsilon' \lfs(p)$. To obtain $E$ as an $\epsilon$-sample, we set $\epsilon'(\epsilon) = (\sqrt{34\epsilon+1}-1)/17$. Observing that $3\epsilon/4 \leq \epsilon'(\epsilon)$ for $\epsilon \leq 1/500$, the returned $\epsilon$-sample satisfies our required sparsity condition with $\sigma \geq 3/4$.

We start by adapting Theorem 6.2 and Lemma 6.4 from~\cite{chazal2006topology} to the setting just described. For $x \in \mathbb{R}^3 \setminus M$, let $\Gamma(x) = \dist(x, \tilde{x}) / \lfs(\tilde{x})$, where $\tilde{x}$ is the closest point to $x$ on $\surf$.

\begin{myCorollary}\label{cor:CL}
 For an $\epsilon$-sample $\sampleset$, with $\epsilon \leq 1/20$, the union of balls $\ballunion$ with $\delta = 2\epsilon$ satisfies:
 \begin{enumerate}
  \item $\surf$ is a deformation retract of $\ballunion$,
  \item $\partial\ballunion$ contains two connected components, each isotopic to $\surf$,
  \item $\Gamma^{-1}([0, a']) \subset U \subset \Gamma^{-1}([0, b'])$, where $a' = \epsilon - 2\epsilon^2$ and $b' \leq 2.5\epsilon$.
 \end{enumerate}
\end{myCorollary}
\begin{proof}
 Theorem 6.2 from~\cite{chazal2006topology} is stated for balls with radii within $[a, b]$ times the $\lfs$. We set $a = b = \delta$ and use $\epsilon \leq 1/20$ to simplify fractions.  This yields the above expressions for $a' = (1 - \epsilon)\delta - \epsilon$ and $b' = \delta/(1 - 2\delta)$. The general condition requires $(1-a')^2 + \big(b'-a'+\delta(1+2b'-a')/(1-\delta)\big)^2 < 1$, as we assume no noise. Plugging in the values of $a'$ and $b'$, we verify that the inequality holds for the chosen range of $\epsilon$.
\end{proof}

Furthermore, we require that each ball $B_i \in \ballset$ contributes one facet to each side of $\partial\ballunion$.  Our sampling conditions ensure that both poles are outside any ball $B_{j} \in \ballset$.

\begin{myLemma}[Disk caps]\label{lem:delta_poles}
 All balls in $\ballset$ have disk caps for $\epsilon \leq 0.066$, $\delta = 2\epsilon$ and $\sigma \geq 3/2$.
\end{myLemma}
\begin{proof}
Fix a sample $p_i$ and let $x$ be one of the poles of $B_i$ and $B_x = \mathbb{B}(c, \lfs(p_i))$ the tangent ball at $p_i$ with $x \in B_x$. Letting $p_j$ be the closest sample to $x$ in $P\setminus\{p_i\}$, we assume the worst case where $\lfs(p_j) \geq \lfs(p_i)$ and $p_j$ lies on $\partial B_x$. To simplify the calculations, take $\lfs(p_i) = 1$ and let $\ell$ denote $\dist(p_i, p_j)$. As $\lfs$ is 1-Lipschitz, we get $\lfs(p_j) \leq 1 + \ell$. By the law of cosines, $\dist(p_j, x)^2 = \dist(p_i, p_j)^2 + \dist(p_i, x)^2 - 2 \dist(p_i, p_j) \dist(p_i, x) \cos(\phi)$, where $\phi = \angle p_j p_i c$. Letting $\theta = \angle p_i c p_j$, observe that $\cos(\phi) = \sin(\theta/2) = \ell/2$. To enforce $x \notin B_j$, we require $\dist(p_j, x) > \delta \lfs(p_j)$, which is equivalent to $\ell^2 + \delta^2 - \delta \ell^2 > \delta^2(1 + \ell)^2$. Simplifying, we get $\ell > 2\delta^2 / (1-\delta-\delta^2)$ where sparsity guarantees $\ell > \sigma \epsilon$. Setting $\sigma \epsilon > 2\delta^2 / (1-\delta-\delta^2)$ we obtain $4\sigma\epsilon^2 + (8+2\sigma)\epsilon -\sigma < 0$, which requires $\epsilon < 0.066$ when $\sigma \geq 3/4$.
\end{proof}

\Cref{cor:CL} together with \Cref{lem:delta_poles} imply that each $\partial\ball_i$ is decomposed into a covered region $\partial B_i \cap \cup_{j\neq i} B_j$, the \emph{medial band}, and two uncovered caps $\partial B_i \setminus \cup_{j\neq i} B_j$, each containing one pole. Recalling that seeds arise as pairs of intersection points between the boundaries of such balls, we show that seeds can be classified correctly as either inside or outside $\surf$.

\begin{myCorollary}\label{cor:seed_necklace}
 If a seed pair lies on the same side of $\surf$, then at least one seed is covered.
\end{myCorollary}
\begin{proof}
 Fix such a seed pair $\partial B_i \cap \partial B_j \cap \partial B_k$ and recall that $\surf \cap \partial B_i$ is contained in the medial band on $\partial B_i$. Now, assume for contradiction that both seeds are uncovered and lie on the same side of $\surf$. It follows that $B_j \cap B_k$ intersects $B_i$ away from its medial band, a contradiction to \Cref{cor:CL}.
\end{proof}

\Cref{cor:CL} guarantees that the medial band of $\ball_i$ is a superset of $\Gamma^{-1}([0, a']) \cap \partial B_i$, which means that all seeds $s_{ijk}$ are at least $a'\lfs(\tilde{s}_{ijk})$ away from $\surf$. It will be useful to bound the elevation of such seeds above $T_{p_i}$, the \textit{tangent plane} to $\surf$ at $p_i$.

\begin{myLemma}\label{lem:seed_elevation}
 For a seed $s \in \partial B_i$, $\theta_s = \angle sp_is' \geq 29.34^\circ$ and $
\theta_s > \frac{1}{2}-5\epsilon$, where $s'$ is the projection of $s$ on $T_{p_i}$, implying $\dist(s, s') \geq h^\bot_s \delta \lfs(p_i)$, with $h^\bot_s > 0.46$ and $h^\bot_s > \frac{1}{2} - 5\epsilon$.
\end{myLemma}
\begin{proof}
 Let $\lfs(p_i) = 1$ and $B_s = \mathbb{B}(c, 1)$ be the tangent ball at $p_i$ with $s \notin B_s$; see~\Cref{fig:seed_elevation}. Observe that $\dist(s, \surf) \leq \dist(s, x)$, where $x = \overline{sc} \cap \partial B_s$. By the law of cosines, $\dist(s, c)^2 = \dist(p_i,c)^2 + \dist(p_i, s)^2 - 2\dist(p_i, c)\dist(p_i,s)\cos(\pi/2+\theta_s) = 1 + \delta^2 + 2\delta\sin(\theta_s)$. We may write\footnote{Define $f(u, v) = \sqrt{1 + u^2+2uv} - (1 + u^2/2 + uv)$ and observe that $f(u, -u/2) = 0$ is the only critical value of $f(u, .)$. As $\partial^2 f/\partial v^2 \leq 0$ for $(u, v) \in \mathbb{R} \times [-1, 1]$, we get that $f(u, v) \leq 0$ in this range.} $\dist(s, c) \leq 1 + \delta^2/2 + \delta\sin(\theta_s)$. It follows that $\dist(s,x) \leq \delta^2/2 + \delta\sin(\theta_s)$. As $\lfs$ is 1-Lipschitz and $\dist(p_i, x) \leq \delta$, we get $1-\delta \leq \lfs(x) \leq 1+\delta$. There must exist a sample $p_j$ such that $\dist(x, p_j) \leq \epsilon \lfs(x) \leq \epsilon(1+\delta)$. Similarly, $\lfs(p_j) \geq (1-\epsilon(1+\delta))(1-\delta)$. By the triangle inequality, $\dist(s, p_j) \leq \dist(s,x) + \dist(x,p_j) \leq \delta^2/2 + \delta\sin(\theta_s) + \epsilon(1+\delta)$. Setting $\dist(s, p_j) < \delta(1-\delta)(1-\epsilon(1+\delta))$ implies $\dist(s,p_j) < \delta\lfs(p_j)$, which shows that for small values of $\theta_s$, $s$ cannot be a seed and $p_j \neq p_i$. Substituting $\delta=2\epsilon$, we get $\theta_s \geq \sin^{-1}{(2\epsilon^3-5\epsilon+1/2)} \geq 29.34^\circ$ and $\theta_s > 1/2 - 5\epsilon$.
\end{proof}

We make frequent use of the following bound on the distance between related samples.
\begin{myClaim}\label{clm:edge_bound}
If $B_i \cap B_j \neq \emptyset$, then $\dist(p_i, p_j) \in [\kappa_\epsilon, \kappa \delta] \cdot \lfs(p_i)$, with $\kappa = 2/(1-\delta)$ and $\kappa_\epsilon = \sigma \epsilon / ( 1 + \sigma \epsilon)$.
\end{myClaim}
\begin{proof}
The upper bound comes from $\dist(p_i, p_j) \leq r_i + r_j$ and $\lfs(p_j) \leq \lfs(p_i) + \dist(p_i, d_j)$ by 1-Lipschitz,
and the lower bound from $\lfs(p_i) - \dist(p_i, d_j) \leq \lfs(p_j)$ and the sparsity.
\end{proof}

Bounding the circumradii is the culprit behind why we need such small values of $\epsilon$. 

\begin{myLemma}\label{lem:small_guides}
The circumradius of a guide triangle $t_{ijk}$ is at most $\varrho_f \cdot \delta \lfs(p_i)$, where $\varrho_f < 1.38$, and at most $\overline{\varrho}_f \cdot \dist(p_i, p_j)$ where $\overline{\varrho}_f < 3.68$.
\end{myLemma}
\begin{proof}
Let $p_i$ and $p_j$ be the triangle vertices with the smallest and largest $\lfs$ values, respectively. From \Cref{clm:edge_bound}, we get $\dist(p_i, p_j) \leq \kappa \delta \lfs(p_i)$. It follows that $\lfs(p_j) \leq (1 + \kappa \delta) \lfs(p_i)$. As $t_{ijk}$ is a guide triangle, we know that it has a pair of intersection points $\partial B_i \cap \partial B_j \cap \partial B_k$. Clearly, the seed is no farther than $\delta\lfs(p_j)$ from any vertex of $t_{ijk}$ and the orthoradius of $t_{ijk}$ cannot be bigger than this distance.

Recall that the weight $w_i$ associated with $p_i$ is $\delta^2\lfs(p_i)^2$. We shift the weights of all the vertices of $t_{ijk}$ by the lowest weight $w_i$, which does not change the orthocenter. With that $w_j - w_i = \delta^2(\lfs(p_j)^2-\lfs(p_i)^2) \leq \delta^2\lfs(p_i)^2((1+\kappa\delta)^2-1) = \kappa\delta^3\lfs(p_i)^2(\kappa\delta+2)$. On the other hand, sparsity ensures that the closest vertex in $t_{ijk}$ to $p_j$ is at distance at least $N(p_j) \geq \sigma\epsilon\lfs(p_j) \geq \sigma\epsilon(1-\kappa\delta)\lfs(p_i)$. Ensuring $\alpha^2 \leq (w_j - w_i)/N(p_i)^2 \leq \kappa\delta^3(2+\kappa\delta)/(\sigma^2\epsilon^2(1-\kappa\delta)^2) \leq 1/4$ suffices to bound the circumradius of $t_{ijk}$ by $c_{rad} = 1/\sqrt{1-4\alpha^2}$ times its orthoradius, as required by Claim 4 in~\cite{exudation}. Substituting $\delta=2\epsilon$ and $\sigma \geq 3/4$ we get $\alpha^2 \leq 78.97\epsilon$, which corresponds to $c_{rad} < 1.37$. It follows that the circumradius is at most $c_{rad}\delta\lfs(p_j) \leq c_{rad}(1 + \kappa \delta)\delta\lfs(p_i) < 1.38 \delta \lfs(p_i)$.

For the second statement, observe that $\lfs(p_i) \geq (1 - \kappa\delta)\lfs(p_j)$ and the sparsity condition ensures that the shortest edge length is at least $\sigma \epsilon \lfs(p_i) \geq \sigma\epsilon(1-\kappa\delta)\lfs(p_j)$. It follows that the circumradius is at most $\frac{\delta c_{rad}}{\sigma\epsilon(1 - \kappa\delta)} < 3.68$ times the length of any edge of $t_{ijk}$.
\end{proof}

Given the bound on the circumradii, we are able to bound the deviation of normals.

\begin{myLemma}\label{lem:guide_normals}
If $t_{ijk}$ is a guide triangle, then (1) $\angle_a (n_{p_i}, n_{p_j}) \leq \eta_s \delta < 0.47^\circ$, with $\eta_s < 2.03$, and (2) $\angle_a (n_t, n_{p_i}) \leq \eta_t \delta < 1.52^\circ$, with $\eta_t < 6.6$, where $n_{p_i}$ is the line normal to $\surf$ at $p_i$ and $n_t$ is the normal to $t_{ijk}$. In particular, $t_{ijk}$ makes an angle at most $\eta_t\delta$ with $T_{p_i}$.
\end{myLemma}
\begin{proof}
\Cref{clm:edge_bound} implies $\dist(p_i, p_j) \leq \kappa \delta \lfs(p_i)$ and (1) follows from the Normal Variation Lemma~\cite{DBLP:journals/corr/AmentaD14} with $\rho = \kappa \delta < 1/3$ yielding $\angle_a (n_{p_i}, n_{p_j}) \leq \kappa \delta / (1 - \kappa\delta)$. Letting $R_t$ denote the circumradius of $t$, \Cref{lem:small_guides} implies that the $R_t \leq \varrho_f \cdot \delta \lfs(p_i) \leq \lfs(p_i) / \sqrt{2}$ and the Triangle Normal Lemma \cite{Dey:2006:CSR:1196751} implies $\angle_a (n_{p^\ast}, n_t) < 4.57\delta < 1.05^\circ $, where $p^\ast$ is the vertex of $t$ subtending a maximal angle in $t$. Hence, $\angle_a (n_{p_i}, n_t) \leq \angle_a (n_{p_i}, n_{p^\ast}) + \angle_a (n_{p^\ast}, n_t)$.
\end{proof}

Towards establishing homeomorphism, the next lemma on the monotonicity of distance to the nearest seed is critical. First, we show that the nearest seeds to any surface point $x \in \surf$ are generated by nearby samples.

\begin{myLemma}\label{lem:nearby_seed}
The nearest seed to $x \in \surf$ lies on some $\partial B_i$ where $\dist(x, p_i) \leq 5.03\cdot\epsilon\lfs(x)$. Consequently, $\dist(x, p_i) \leq 5.08\cdot\epsilon\lfs(p_i)$.
\end{myLemma}
\begin{proof}
In an $\epsilon$-sampling, there exists a $p_a$ such that $\dist(x, p_a) \leq \epsilon\lfs(x)$, where $\lfs(p_a) \leq (1+\epsilon)\lfs(x)$. The sampling conditions also guarantee that there exists at least one seed $s_a$ on $\partial \ball_a$. By the triangle inequality, we get that $\dist(x, s_a) \leq \dist(x, p_a) + \dist(p_a, s_a) \leq \epsilon\lfs(x) + \delta\lfs(p_a) \leq \epsilon(1 + 2(1+\epsilon))\lfs(x) = \epsilon(2\epsilon+3)\lfs(x)$.

We aim to bound $\ell$ to ensure $\forall\samp_i$ s.t. $\dist(x, p_i) = \ell\cdot\epsilon\lfs(x)$, the nearest seed to $x$ cannot lie on $B_i$. Note that in this case, $(1-\ell\epsilon)\lfs(x) \leq \lfs(p_i) \leq (1+\ell\epsilon)\lfs(x)$. Let $s_i$ be any seed on $B_i$. It follows that $\dist(x, s_i) \geq \dist(x, p_i) - \dist(p_i, s_i) \geq \ell\cdot\epsilon\lfs(x) - 2\epsilon\lfs(p_i) \geq \epsilon\big((1-2\epsilon)\ell-2\big)\lfs(x)$.

Setting $\epsilon\big((1-2\epsilon)\ell-2\big)\lfs(x) \geq \epsilon(2\epsilon+3)\lfs(x)$ suffices to ensure $\dist(x, s_i) \geq \dist(x, s_a)$, and we get $\ell \geq (2\epsilon+5)/(1-2\epsilon)$. Conversely, if the nearest seed to $x$ lies on $B_i$, it must be the case that $\dist(x, p_i) \leq \ell\epsilon\lfs(x)$. We verify that $\ell\epsilon = \epsilon(2\epsilon+5)/(1-2\epsilon) < 1$ for any $\epsilon < 0.13$. It follows that $\dist(x, p_j) \leq \ell\epsilon/(1-\ell\epsilon)\lfs(p_i)$.
\end{proof}

\begin{myLemma}\label{lem:monotonicity}
 For any normal segment $N_x$ issued from $x \in \surf$, the distance to $\seedset^\uparrow$ is either strictly increasing or strictly decreasing along $\Gamma^{-1}([0, 0.96\epsilon]) \cap N_x$. The same holds for $\seedset^\downarrow$.
\end{myLemma}
\begin{proof}
 Let $n_x$ be the outward normal and $T_x$ be the tangent plane to $\surf$ at $x$. By \Cref{lem:nearby_seed}, the nearest seeds to $x$ are generated by nearby samples. Fix one such nearby sample $p_i$. For all possible locations of a seed $s \in \seedset^\uparrow \cap \partial B_i$, we will show a sufficiently large lower bound on $\langle s-s'', n_x \rangle$, where $s''$ the projection of $s$ onto $T_x$.
 
 Take $\lfs(p_i) = 1$ and let $B_s = \mathbb{B}(c, 1)$ be the tangent ball to $\surf$ at $p_i$ with $s \in B_s$. Let $A$ be the plane containing $\{p_i, s, x\}$. Assume in the worst case that $A \bot T_{p_i}$ and $x$ is as far as possible from $p_i$ on $\partial B_s \cap T_{p_i}$. By \Cref{lem:nearby_seed}, $\dist(p_i, x) \leq 5.08\epsilon$ and it follows that $\theta_x = \angle (n_x, n_{p_i}) \leq 5.08\epsilon/(1-5.08\epsilon) \leq 5.14\epsilon$. This means that $T_x$ is confined within a $(\pi/2-\theta_x)$-cocone centered at $x$. Assume in the worst case that $n_x$ is parallel to $A$ and $T_x$ is tilted to minimize $\dist(s, s'')$; see \Cref{fig:monotonicity}.
 
 Let $T'_x$ be a translation of $T_x$ such that $p_i \in T'_x$ and denote by $x'$ and $s'$ the projections of $x$ and $s$, respectively, onto $T'_x$. Observe that $T'_x$ makes an angle $\theta_x$ with $T_{p_i}$. From the isosceles triangle $\triangle p_icx$, we get that $\theta'_x \leq 1/2\angle p_icx = \sin^{-1}{5.08\epsilon/2} \leq 2.54\epsilon$. Now, consider $\triangle p_ixx'$ and let $\phi = \angle xp_ix'$. We have that $\phi = \theta_x + \theta'_x \leq 2.54\epsilon + \delta/(1-\delta) \leq 4.55\epsilon$. Hence, $\sin(\phi) \leq 4.55\epsilon$ and $\dist(x, x') \leq 5.08\epsilon \sin(\phi) \leq 0.05\epsilon$. On the other hand, we have that $\angle sp_is' = \psi \geq \theta_s - \theta_x$ and $\dist(s, s') \geq \delta \sin\psi$, where $\theta_s \geq 1/2 - 5\epsilon$ by \Cref{lem:seed_elevation}. Simplifying we get $\sin(\psi) \geq 1/2 - 10.08\epsilon$. The proof follows by evaluating $\dist(s, s'') = \dist(s, s') - \dist(x, x')$.
\end{proof}

\begin{figure}
\centering
\begin{subfigure}[b]{0.25\columnwidth}\centering
  \includegraphics[width=0.7\hsize]{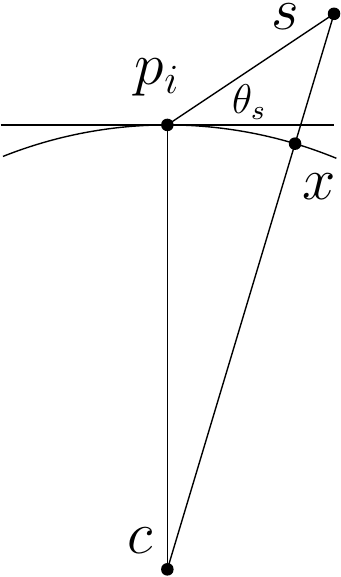}
  \caption{Seed elevation $\theta_s$.}
   \label{fig:seed_elevation}
\end{subfigure}
\quad
\begin{subfigure}[b]{0.4\columnwidth}\centering
  \includegraphics[width=0.95\hsize, trim={1.6cm 0 1.2cm 0}, clip]{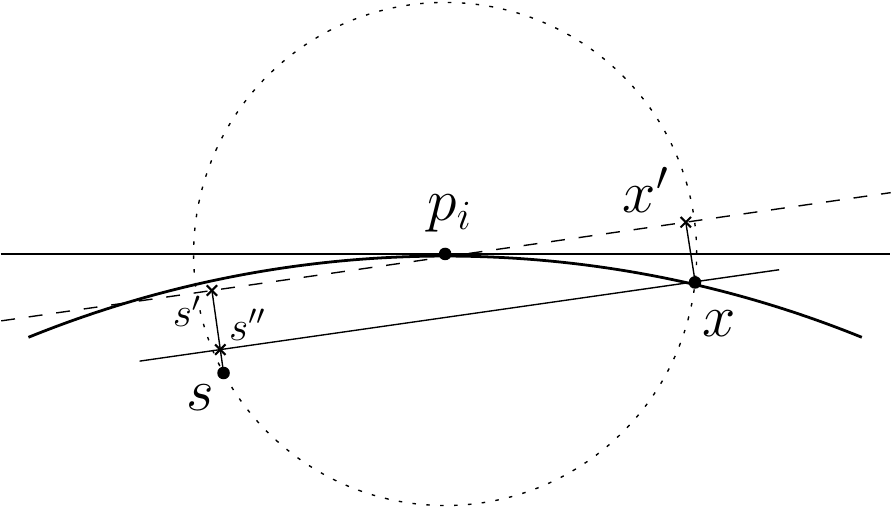}
  \caption{Bounding seed height above $T_x$.}
  \label{fig:monotonicity}
\end{subfigure}
\quad
\begin{subfigure}[b]{0.25\columnwidth}\centering
  \includegraphics[width=0.7\hsize]{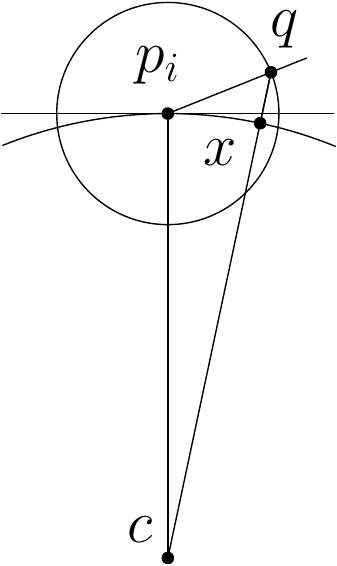}
  \caption{Bounding $\dist(q, \surf)$.}
  \label{fig:coconeQP}
\end{subfigure}
 \caption{Constructions used for (a) \Cref{lem:seed_elevation}, (b) \Cref{lem:monotonicity} and (c) \Cref{thm:approximation}.}
  \label{fig:cocone}
\end{figure}

\begin{theorem}\label{thm:approximation}
For every $x\in\surf$ with closest point $q\in\newsurf$, and for every $q\in\newsurf$ with closest point $x\in\surf$, we have $\|xq\| < h_t \cdot\epsilon^2\lfs(x)$, where $h_t < 30.52$. For $\epsilon < 1/500$, $h_t \cdot\epsilon^2 < 0.0002$. Moreover, the restriction of the mapping $\pi$ to $\newsurf$ is a homeomorphism and $\newsurf$ and $\surf$ are ambient isotopic. Consequently, $\hat{\mathcal{O}}$ is ambient isotopic to $\mathcal{O}$ as well.
\end{theorem}
\begin{proof}
Fix a sample $\samp_i \in \sampleset$ and a surface point $x \in \surf \cap \ball_i$. We consider two cocones centered at $x$: a $p$-cocone contains all nearby surface points and a $q$-cocone contains all guide triangles incident at $p_i$. By \Cref{thm:sphere_sandwich}, all reconstruction facets generated by seeds on $\ball_i$ are sandwiched in the $q$-cocone.

\Cref{lem:guide_normals} readily provides a bound on the q-cocone angle as $\gamma \le \eta_t \delta$.
In addition, since $\dist(\samp_i, x) \leq \delta \lfs(\samp_i)$, we can bound the $p$-cocone angle as $\theta \leq 2\sin^{-1}{(\delta/2)}$ by Lemma 2 in~\cite{amenta1999surface}.
We utilize a mixed $pq$-cocone with angle $\omega = \gamma/2 + \theta/2$, obtained by gluing the lower half of the $p$-cocone with the upper half of the $q$-cocone.

Let $q\in\newsurf$ and consider its closest point $x\in\surf$. Again, fix $\samp_i \in \sampleset$ such that $x \in \ball_i$; see \Cref{fig:coconeQP}.
By sandwiching, we know that any ray through $q$ intersects at least one guide triangle, in some point $y$, after passing through $x$.
Let us assume the worst case that $y$ lies on the upper boundary of the $pq$-cocone.
Then, $\dist(q,x) \leq \dist(y, y') = h = \delta \sin(\omega) \lfs(p_i)$, where $y'$ is the closest point on the lower boundary of the $pq$-cocone point to $q$.
We also have that, $\dist(p_i,x) \le \cos(\omega)\delta\lfs(p_i) \le \delta \lfs(p_i)$, and 
since $\lfs$ is 1-Lipschitz, $\lfs(p_i)  \le \lfs(x) / (1 - \delta)$.
Simplifying, we write $\dist(q,x) < \delta\omega/(1-\delta) \cdot \lfs(x) < h_t \epsilon^2 \lfs(x)$.

With $\dist(q, x) \leq 0.55 \epsilon \lfs(x)$, \Cref{lem:monotonicity} shows that the normal line from any $p \in \surf$ intersects $\newsurf$ exactly once close to the surface. It follows that for every point $x \in \surf$ with closest point $q \in \newsurf$, we have $\dist(x, q) \leq \dist(x, q')$ where $q' \in \newsurf$ with $x$ its closest point in $\surf$. Hence, $\dist(x, q) \leq h_t \epsilon^2 \lfs(x)$ as well.

Building upon \Cref{lem:monotonicity}, as a point moves along the normal line at $x$, it is either the case that the distance to $\seedset^\uparrow$ is decreasing while the distance to $\seedset^\downarrow$ is increasing or the other way around. It follows that these two distances become equal at exactly one point on the Voronoi facet above or below $x$ separating some seed $s^\uparrow \in \seedset^\uparrow$ from another seed $s^\downarrow \in \seedset^\downarrow$. Hence, the restriction of the mapping $\pi$ to $\newsurf$ is a homeomorphism.

This shows that $\newsurf$ and $\surf$ homeomorphic. Recall that \Cref{cor:CL}(3) implies $\ballunion$ is a \textit{topological thickening}~\cite{CHAZAL2005390} of $\surf$. In addition, \Cref{thm:sphere_sandwich} guarantees that $\newsurf$ is embedded in the interior of $\ballunion$, such that it separates the two surfaces comprising $\partial\ballunion$. These three properties imply $\newsurf$ is isotopic to $\surf$ in $\ballunion$ by virtue of Theorem 2.1 in~\cite{CHAZAL2005390}. Finally, as $\newsurf$ is the boundary of $\hat{\mathcal{O}}$ by definition, it follows that $\hat{\mathcal{O}}$ is isotopic to $\mathcal{O}$ as well.
\end{proof}

\section{Quality guarantees and output size} \label{sec:quality}

We establish a number of quality guarantees on the output mesh. The main result is an upper bound on the \emph{fatness} of all Voronoi cell.
See~\Cref{sec:quality_proofs} for the proofs.

Recall that fatness is the outradius to inradius ratio, where the outradius is the radius of the smallest enclosing ball, and the inradius is the radius of the largest enclosed ball. The good quality of guide triangles allows us to bound the inradius of Voronoi cells.

\begin{restatable}[]{myLemma}{minAltitude}\label{lem:minAltitude}
For all guide triangles $t_{ijk}$:
(1) Edge length ratios are bounded: $\ell_k / \ell_j \le \kappa_\ell = \frac{2\delta}{1-\delta} \frac{\sigma \epsilon}{ 1 +  \sigma \epsilon}.$
(2) Angles are bounded: $\sin(\theta_i) \ge 1 / (2 \overline{\varrho}_f)$ implying $\theta_i \in (7.8^\circ,165^\circ)$.
(3) Altitudes are bounded: the altitude above $e_{ij}$ is at least $\alpha_t |e_{ij}|$, where $\alpha_t = 1 / 4\overline{\varrho}_f > 0.067$.
\end{restatable}

Observe that a guide triangle is contained in the Voronoi cell of its seed, even when one of the guides is covered. Hence, the tetrahedron formed by the triangle together with its seed lies inside the cell,
and the cell inradius is at least the tetrahedron inradius.

\begin{restatable}[]{myLemma}{tetInradius}\label{lem:tet_inradius}
For seeds $s_{ijk} \in \seedset^\uparrow \cup \seedset^\downarrow$, the inradius of the Voronoi cell is at least 
$\varrho_v \delta \cdot \lfs(p_i)$ with $\varrho_v = \hat{h}_s/(1 + \frac{3}{2 \sigma \overline{\varrho}_f }) > 0.3$ and $\hat{h}_s \geq \frac{1}{2} - (5 + 2\eta_t) \epsilon$.
\end{restatable}

To get an upper bound on cell outradii, we must first generate seeds interior to $\vol$.
We consider a simple algorithm for generating $\seedset^\ddarrow$ based on a standard octree over $\vol$.
For sizing, we extend \lfs\ beyond $\surf$, using the point-wise maximal 1-Lipschitz extension $\lfs(x) = \inf_{p\in\surf}( \lfs(p) + \dist(x,p) )$~\cite{LipschitzExtension}.
An octree box $\square$ is refined if the length of its diagonal is greater than $2\delta \cdot \lfs(c)$, where $c$ is the center of $\square$.
After refinement terminates, we add an interior seed at the center of each empty box, and do nothing with boxes containing one or more guide seeds. Applying this scheme, we obtain the following.
%
%

\begin{restatable}{myLemma}{interiorAR}
The fatness of interior cells is at most $\frac{8 \sqrt{3} (1+\delta)}{1-3\delta} < 14.1$.
\end{restatable}

\begin{restatable}{myLemma}{boundaryAR}
The fatness of boundary cells is at most $\frac{4 (1+\delta)}{(1-3\delta)(1-\delta)^2 \varrho_v} < 13.65$.
\end{restatable}

As the integral of $\lfs^{-3}$ is bounded over a single cell, it effectively counts the seeds.

\begin{restatable}[]{myLemma}{meshSize}
$|\seedset^\ddarrow| \le 18\sqrt{3}/\pi \cdot \epsilon^{-3} \int_\vol \lfs^{-3}$.
\end{restatable}

\section{Conclusions}

We have analyzed an abstract version of the VoroCrust algorithm for volumes bounded by smooth surfaces. We established several guarantees on its output, provided the input samples satisfy certain conditions. In particular, the reconstruction is isotopic to the underlying surface and all 3D Voronoi cells have bounded fatness, i.e., outradius to inradius ratio. The triangular faces of the reconstruction have bounded angles and edge-length ratios, except perhaps in the presence of slivers. In a forthcoming paper~\cite{vorocrust_02_surface_sampling}, we describe the design and implementation of the complete VoroCrust algorithm, which generates conforming Voronoi meshes of realistic models, possibly containing sharp features, and produces samples that follow a natural sizing function and ensure output quality.

For future work, it would be interesting to ensure both guides are uncovered, or both covered. The significance would be that no tetrahedral slivers arise and no Steiner points are introduced. Further, the surface reconstruction would be composed entirely of guide triangles, so it would be easy to show that triangle normals converge to surface normals as sample density increases. Alternatively, where Steiner points are introduced on the surface, it would be helpful to have conditions that guaranteed the triangles containing Steiner points have good quality. In addition, the minimum edge length in a Voronoi cell can be a limiting factor in certain numerical solvers. Post-processing by mesh optimization techniques~\cite{cgf.13256,sieger2010optimizing} can help eliminate short Voronoi edges away from the surface. Finally, we expect that the abstract algorithm analyzed in this paper can be extended to higher dimensions.

\bibliography{./references}

\clearpage
\appendix

\section{Sandwich analysis} \label{sec:sandwich_analysis}
This appendix provides the proof of~\Cref{thm:sphere_sandwich}. In what follows, we assume all sample spheres have disk caps; see \Cref{sec:disk_caps}. By \Cref{cor:sample_reconstruction}, each $\samp_i \in \sampleset$ appears as a vertex in the surface reconstruction $\newsurf$. For an appropriate sampling, $\newsurf$ is watertight and the set of reconstruction facets containing $\samp_i$ is a topological disk with $\samp_i$ in its interior.
We seek to characterize the angular orientation of the fan-facets around a sample $\samp_i$,
namely that they lie sandwiched between upper and lower guide triangles.
For orientation, it suffices to consider only the seeds  on $\sphere_i,$ because 
these are the seeds whose cells contain $\samp_i$.
By suppressing other seeds, the fan-facets are all triangles extending radially from $\samp_i$ to infinity.
To argue about the actual reconstruction $\newsurf$, those other seeds are introduced later.

We make use of technical arguments in spherical geometry concerning the spherical arcs formed by the intersections of $\sphere_i$ with extended reconstruction facets, and extended guide triangles. These arcs are great circle arcs between the two points where $\sphere_i$ intersects the two rays from $\samp_i$ that bound the extended triangle (facet). 

For a reconstruction facet $t \in \newsurf$ incident to $\samp_i \in \sampleset$, let $\hat{t}$ be its radial extension from $p_i$ to infinity. Similarly, we denote by $\hat{t}_{ijk}$ the radial extension of the guide triangle $t_{ijk}$. An \emph{extended fan-facet point} is any point $m \in \hat{t} \cap \sphere_i$. We will show that such points $m$ lie sandwiched between extended guide triangles incident to $\samp_i$. Since $t$ is the intersection of the Voronoi cells of one upper and one lower seed, $m \in \hat{t}$ is the center of a \textit{Voronoi ball} $\ball_m$ having those two seeds on its boundary and no seed in its interior. The intersection of the empty ball $\ball_m$ with $\sphere_1$ is the \textit{spherical disk} $D_m \subset \sphere_1$, and also has an interior $\overline{D}_m$ empty of seeds. The following technical lemma is essential to the sandwiching arguments.

\begin{apxLemma}[Path $m$ to $\aseed$ misses $K$]
\label{lem:mgU}
If all sample balls have disk caps, then for any $m$ on an extended reconstruction facet $\hat{t}$ having $\samp_i \in \sampleset$ as a vertex, the shorter great circle arc on $\sphere_i$ from $m$ to the lower (upper) seed of $t$ does not pass through the uppler (lower) cap.
\end{apxLemma}
\begin{proof}
We use numerals to index the unusually large number of samples and seeds involved in the configurations at hand.
Without loss of generality let $i = 1$ and fix one such facet $t$. We consider up to seven spheres $\sphere_{\{2, \dots, 8\}}$, not necessarily distinct, intersecting $\sphere_1$.
Denote the upper and lower seeds associated with $t$ by $\usi{123}$ and $\lsi{145}$. We show that $\overarc{m\lsi{145}}$ does not pass through $\ucap_1$; the case of $\overarc{m\usi{123}}$ and $\lcap_1$ is similar. The proof follows by examining great circle arcs on $\sphere_1$ restricted within $D_m$. By the definition of disk caps, $\ucap_1$ is a topological disk and $\urim_1$ forms a closed path. In addition $\urim_1$ is disjoint from $\lrim_1$, as $\ucap_1$ and $\lcap_1$ are separated by a medial band. Being a lower seed, $\lsi{145}$ lies on $\lrim_1$.

Suppose for contradiction that $\overarc{m\lsi{145}}$ goes through $\ucap_1$ and let $x \in \overarc{m\lsi{145}} \subset \partial D_m$ be the last point where $\overarc{m\lsi{145}}$ crosses $\urim_1$ in the direction from the interior of $\ucap_1$ to its exterior. 
We have that $x$ lies on an uncovered arc $\overarc{\usi{167}\usi{168}} \subset \partial\ucap_1$ on the guide circle $\cir{16}$ with seeds $\usi{167}$ and $\usi{168}$ as endpoints. \Cref{fig:sandwich_v_b} depicts a hypothetical path $\overarc{m\lsi{145}}$. Since $D_m$ is empty, $\usi{167}$ and $\usi{168}$ both lie outside $D_m$. In addition, $\lsi{145}$ cannot lie inside $\cir{16}$, as $\lsi{145}$ is an uncovered seed and $\cir{16}$ bounds the region on $\sphere_1$ which is covered by $\ball_6$; this impossible configuration is illustrated by the dashed circle in \Cref{fig:sandwich_v_b}. It would follow that $\overarc{m\lsi{145}}$ must cross $\overarc{\usi{167}\usi{168}}$ a second time to reach $\lsi{145}$, a contradiction to $x$ being the last point where $\overarc{m\lsi{145}}$ intersects $\urim_1$.
(Special cases include $\ucap_1$ and $\lcap_1$ meeting at a point, or $x$ being a point of tangency between $\cir{16}$ and $\overarc{m\lsi{145}}$.)
\end{proof}

\begin{figure}[H]
\centering
\begin{subfigure}[b]{0.4\columnwidth}\centering
  \includegraphics[width=0.95\textwidth]{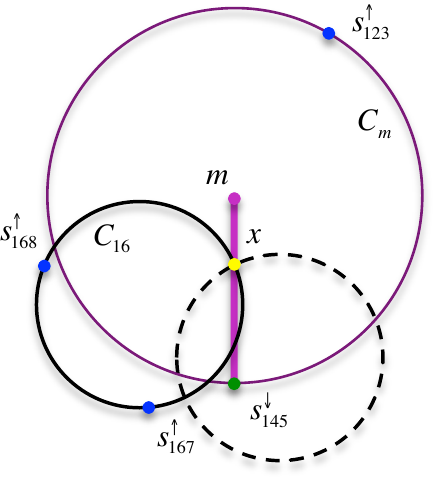} 
  \caption{\Cref{lem:mgU}: arcs from extended reconstruction $m$ to $\ls$ cannot cross $\ucap$.}
  \label{fig:sandwich_v_b}
\end{subfigure}
\quad
\begin{subfigure}[b]{0.55\columnwidth}\centering
  \includegraphics[width=0.95\textwidth]{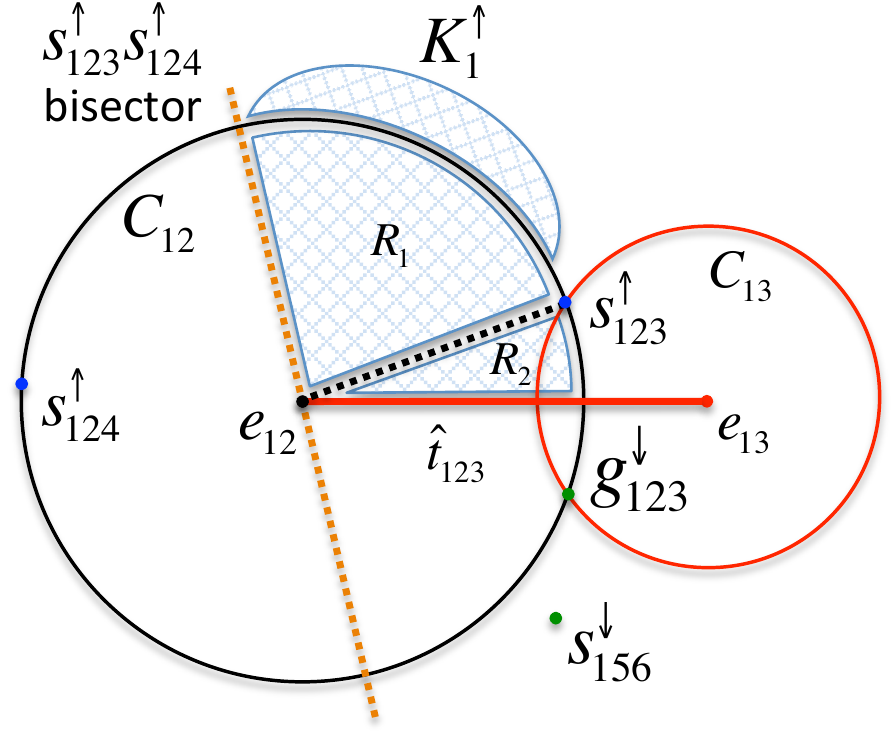} 
  \caption{\Cref{lem:sandwich_vertex}: $m$'s hypothetical locations above $\hat{t}_{123}$: upper cap $\ucap_1$, $R_1$ and $R_2$.}
  \label{fig:sandwich_v_c}
\end{subfigure}
\\
\begin{subfigure}[b]{0.3\columnwidth}\centering
  \includegraphics[width=1.0\textwidth]{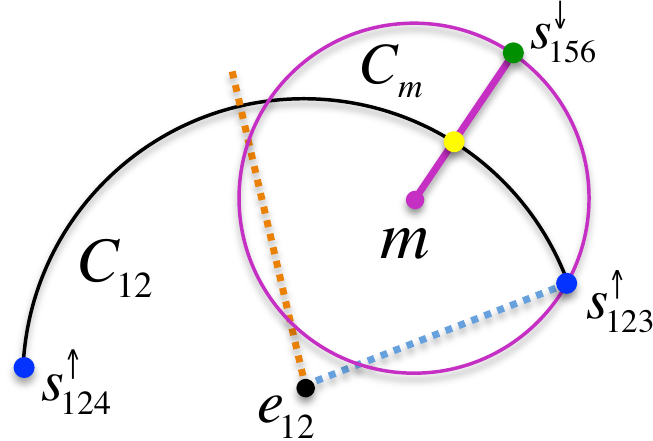} 
  \caption{\Cref{lem:sandwich_vertex}: $m \notin R_1$.}
  \label{fig:sandwich_v_z}
\end{subfigure}
\quad
\begin{subfigure}[b]{0.34\columnwidth}\centering
  \includegraphics[width=1.0\textwidth]{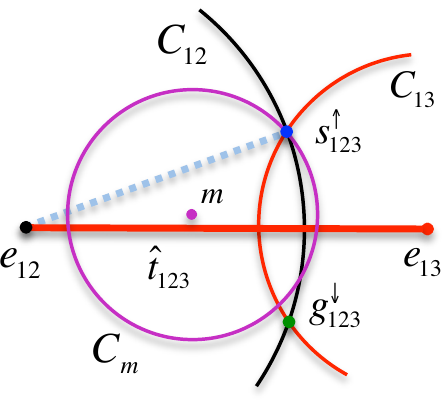} 
  \caption{\Cref{lem:sandwich_vertex}: $m \notin R_2$ (1).}
  \label{fig:sandwich_v_d}
\end{subfigure}
\quad
\begin{subfigure}[b]{0.29\columnwidth}\centering
  \includegraphics[width=1.0\textwidth]{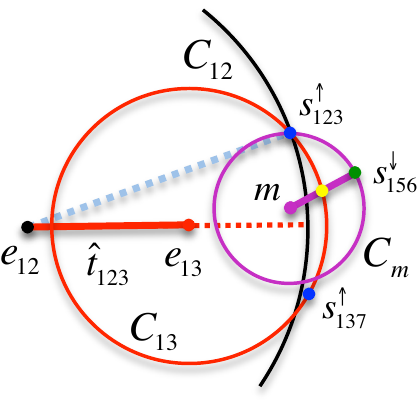} 
  \caption{\Cref{lem:sandwich_vertex}: $m \notin R_2$ (2).}
  \label{fig:sandwich_v_g}
\end{subfigure}
   \caption{Fan-facet sandwich lemmas.}
   \label{fig:sandwiches}
\end{figure}
     
A reconstructed fan-facet is \emph{sandwiched} at $\samp_i$ if the intersection of its radial extension with $\sphere_i$ lies in the region bounded by the extension of guide triangles incident to $\samp_i$.


\begin{apxLemma}[Sandwich fan-facets]\label{lem:sandwich_vertex}
If all sample balls have disk caps, then all facets in $\newsurf$ with a sample $\samp_i \in \sampleset$ as a vertex are sandwiched between guide triangles.
\end{apxLemma}
\begin{proof}
Take $i = 1$ and consider the spherical arcs on $\sphere_1$ which arise from intersecting $\sphere_1$ with extended guide triangles having $\samp_1$ as a vertex. We partition the subset of $\sphere_1$ above these arcs, and show that no point $m$ of an extended reconstruction facet can lie in any of these partitions; the argument for the subset of $\sphere_1$ below these arcs is analogous. By Lemma~\ref{lem:mgU}, such $m$ may only lie on the medial band.

Consider the spherical arcs bounding the upper cap $\ucap_1$. As in the proof of the prior lemma, we label up to seven other spheres $\sphere_{\{2, \dots, 7\}}$ that intersect $\sphere_1$. Each consecutive pair of such arcs intersect at an upper seed $\usi{\ast}$ and a lower guide $\lxi{*}$ which may or may not be covered. Fix such an arc $\overarc{\usi{123}\usi{124}} \subset \partial\ucap_1$ on the guide circle $\cir{12}$. With a slight abuse of notation, we denote by $e_{1j}$ the intersection of the ray $\overrightarrow{\samp_1\samp_j}$ with $\sphere_1$. Let $D_{1j}$ denote the spherical disk bounded by $\cir{1j}$ and $\overline{D}_{1j}$ its interior. We split $D_{12}$ in half by the great circle arc through $e_{12}$ bisecting $\overarc{\usi{123} \usi{124}}$. Without loss of generality, we consider the half disk containing $\usi{123}$ and argue that no point $m$ on an extended reconstruction facet can lie above the extended upper guide triangle $\hat{t}_{123}$; the case of lower guide triangles is similar.

We partition the half disk above the extended upper guide triangle $\hat{t}_{123}$ into $R_1$, above $\overarc{e_{12}\usi{123}}$ and $R_2$, below $\overarc{e_{12}\usi{123}}$ and above the extended upper guide facet $\hat{t}_{123}$. See \Cref{fig:sandwich_v_c}. Suppose there is some point $m$ on an extended reconstruction facet whose two closest and equidistant seeds are $\usi{123}$ and some lower seed which we denote by $\lsi{156}$. We show that $m$ cannot lie above $\hat{t}_{123}$, neither in $R_1$ nor $R_2$. Observe that if $m = e_{12}$, then it clearly does not lie above $\hat{t}_{123}$. Hence, we assume $m \neq e_{12}$.

Suppose $m \in R_1$ and consider the spherical disk $D_{m}$ on $\sphere_1$, with center $m$ and $\usi{123}$ on its boundary. As no seed is closer to $m$ than $\usi{123}$, $D_{m}$ is empty of other seeds and cannot contain $\usi{124}$ in its interior. Hence, the only portion of $D_{m}$ outside $D_{12}$ is bounded by $\overarc{\usi{123}\usi{124}}$; see \Cref{fig:sandwich_v_z}. In addition, $\lsi{156}$ cannot lie in $\overline{D}_{12}$. It follows that the shorter great circle arc from $m$ to $\lsi{156} \in \partial D_m \setminus \overline{D}_{12}$ must cross $\overarc{\usi{123}\usi{124}} \subset \partial\ucap_1$ into the upper cap $\ucap_1$; a contradiction to \Cref{lem:mgU}. We conclude that $m \notin R_1$.

Suppose $m \in R_2$ and consider the spherical disk $D_{13}$ centered at $e_{13}$ with $\usi{123}$ on its boundary and no seed in its interior.
We have two subcases: either $D_m \subset D_{12} \cup D_{13}$ or not.
Suppose $D_m \subset D_{12} \cup D_{13}$; see \Cref{fig:sandwich_v_d}.
It would follow that the only possible position for $\lsi{156}$ is $\lxi{123}$. But then, $m \in \overarc{e_{12}e_{13}} \subset \hat{t}_{123}$; a contradiction to $m \in R_2$. In the second subcase, let $\usi{137}$ be the other seed on $\cir{13}$ such that $\overarc{\usi{123} \usi{137}} \subset \partial\ucap_1$; see \Cref{fig:sandwich_v_g}. As in the case when $m$ was assumed to lie in $R_1$, the only portion of $D_m$ outside $D_{12} \cup D_{13}$ is bounded by $\overarc{\usi{123} \usi{137}}$ where $\lsi{156}$ cannot lie in $\overline{D}_{12} \cup \overline{D}_{13}$. It follows that the shorter great circle arc from $m$ to $\lsi{156} \in \partial D_m \setminus (\overline{D}_{12} \cup \overline{D}_{13})$, must cross $\overarc{\usi{123} \usi{137}} \subset \partial\ucap_1$ into $\ucap_1$; a contradiction to \Cref{lem:mgU}. As neither subcase can be true, we conclude $m \notin R_2$.
\end{proof}

\sandwiching*
\begin{proof}
Let $V_{ijk}$ denote the Voronoi cell of seed $\as_{ijk}$ and $\sampleset_{ijk} = \{\samp_i, \samp_j, \samp_k\}$. By \Cref{cor:sample_reconstruction}, all three samples $\sampleset_{ijk}$ appear as Voronoi vertices of $V_{ijk}$ and are retained as vertices in $\newsurf$. For an appropriate sampling, $\newsurf$ is watertight and each sample is surrounded by a fan of reconstruction facets. By \Cref{lem:sandwich_vertex}, the reconstruction facets of each such fan are sandwiched between the guide triangles incident on the corresponding sample. In particular, the subset of such reconstruction facets on the boundary of $V_{ijk}$ and incident on samples in $\sampleset_{ijk}$, denoted by $\mathcal{F}_{ijk}$, are sandwiched in this manner. We consider upper seeds; the case of lower seeds is similar. As $V^\uparrow_{ijk}$ is convex, the planes containing each of the facets in $\mathcal{F}^\uparrow_{ijk}$ are tangents to $V^\uparrow_{ijk}$. By \Cref{lem:sandwich_vertex}, the planes containing the subset of guide triangles incident on each of the samples in $\sampleset_{ijk}$, and bounding the reconstruction facets in $\mathcal{F}^\uparrow_{ijk}$ from below, are also tangents to $V^\uparrow_{ijk}$. Hence, all the Voronoi cells of upper seeds lie above the guide triangles incident on samples in $\sampleset_{ijk}$. Similarly, all the Voronoi cells of lower seeds lie below the upper guide triangles incident on the associated triplet of sample points. As $\newsurf$ is the intersection of the Voronoi cells of upper and lower seeds, it follows that $\newsurf$ is sandwiched between upper and lower guide triangles, which constitute $\partial\AlphaS(\sampleset)$.
\end{proof}

\section{Quality bounds} \label{sec:quality_proofs}
This appendix provides the details of the proposed octree refinement for seeding the interior of $\vol$ along with the proofs of the statements in  \Cref{sec:quality}. Namely, we bound the fatness of all cells in the Voronoi mesh $\newvol$ as well as a number of quality measures of guide triangles which constitute the majority of facets in the surface reconstruction $\newsurf$.

\subsection{At the boundary}

Recall that the seeds in $\seedset^\uparrow \cup \seedset^\downarrow \subset \partial\ballunion$ are used to define the surface reconstruction $\newsurf$ as the set of Voronoi facets common to one seed from $\seedset^\uparrow$ and another $\seedset^\downarrow$. Guide triangles contribute the majority of facets in $\newsurf$. Thanks to the sparsity of the sampling, we can derive several quality measures for guide triangles.

\minAltitude*
\begin{proof}
The edge ratio bound is basically a restatement of \Cref{clm:edge_bound}.
Denote by $\ell_i$ and $\theta_i$ the length of the triangle edge opposite to $p_i$ and the angle at vertex $p_i$, respectively. 
\Cref{clm:edge_bound} implies $\ell_k \leq \kappa\delta\lfs(p_i)$ and the sparsity condition guarantees that $\ell_j \geq \kappa_\epsilon \lfs(p_i)$, hence $\ell_i / \ell_k \le \kappa_\ell$ for any pair of edges.

Let $R_{ijk}$ denote $t_{ijk}$'s circumradius.
By the Central Angle Theorem, $\sin(\theta_i) = \ell_i / (2R_{ijk})$, and we also have $R_{ijk} \leq \overline{\varrho}_f\ell_i$ from \Cref{lem:small_guides}.
Hence $\sin(\theta_i) \ge 1 / (2 \overline{\varrho}_f)$.

For the worst case altitude, let the edge under consideration be the longest, $e = \ell_k$, 
and the second longest edge $\ell_j,$ so $\ell_j \ge \ell_k / 2$.
The altitude is then $\sin(\theta_i) \ell_j \ge \ell_k / (4 \overline{\varrho}_f)$.
\end{proof}

Before proceeding to study the decomposition of the interior of $\vol$, we establish a bound on the inradius of Voronoi cells with seeds in $\seedset^\uparrow \cup \seedset^\downarrow$.

\begin{apxCorollary}\label{cor:seed_height}
If $t_{ijk}$ is a guide triangle with associated seed $s$, then $\angle sp_is'' \geq \frac{1}{2}-\eta'_t\epsilon$, where $s''$ is the projection of $s$ on the plane of $t_{ijk}$ and $\eta'_t \leq 5+2\eta_t < 18.18$, implying $\dist(s, s'') \geq \hat{h}_s \delta \lfs(p_i)$ with $\hat{h}_s \geq \frac{1}{2} - \eta'_t \epsilon$.
\end{apxCorollary}
\begin{proof}
 Combining \Cref{lem:seed_elevation} with \Cref{lem:guide_normals}, we have $\angle sp_is'' \geq \angle sp_is' - \angle_a (n_{t_{ijk}}, n_{p_i})$.
\end{proof}

\tetInradius*
\begin{proof}
Fix a seed $s_{ijk}$ and observe that $\{p_i, p_j, p_k\}$ belong to its Voronoi cell. By the convexity of the cell, it follows that the tetrahedron $T = p_ip_jp_ks_{ijk}$ is contained inside it. We establish a lower bound on the cell's inradius by bounding the inradius of $T$. Let $f_i$ denote the facet of $T$ opposite to $p_i$ and $f_0$ denote $t_{ijk}$. Let $A_i$ be the area of $f_i$.

Observe that the incenter $c_T$ divides $T$ into four smaller tetrahedra, one for each facet of $T$, where the distance from $c_T$ to the plane of each facet is equal to the inradius $r$. This allows us to express the volume of $T$ as $V = \sum_{i=0}^{3} rA_i / 3$. Hence, we have that $r = 3V / \sum_i A_i$. We may also express $V$ as $HA_0/3$, where $H$ is the distance from $s_{ijk}$ to the plane of $t_{ijk}$. Substituting for $V$ and factoring out $A_0$, we get that $r = H/(1 + \sum_{i>0}^3A_i/A_0)$.

Triangle area ratios $A_i/A_0$ are bounded because triangle angles are bounded, and edge lengths are bounded by the local feature size.
Consider the edge $e_i=\overline{p_jp_k}$ common to $f_i$ and $t_{ijk}$ and let $\alpha_s$ and $\alpha_p$ be the altitudes of $e_i$ in $f_i$ and $t_{ijk}$, respectively. It follows that $A_i / A_0 = \alpha_s / \alpha_p$.
Note $\alpha_s$ is less than the length of the longest edge of $f_i$.

Hence, assuming that $\lfs(p_j) \geq \lfs(p_k)$, we get that $\alpha_s \le \delta \lfs(p_j)$. On the other hand, the sparsity condition guarantees $\dist(p_j, p_k) \geq \sigma \epsilon\lfs(p_j)$, allowing us to rewrite $\alpha_s \leq \frac{\delta}{\sigma\epsilon}\dist(p_j, p_k)$. 
From \Cref{lem:minAltitude}, we have that $\alpha_p \geq\dist(p_j, p_k) / (4 \overline{\varrho}_f$).
It follows that $A_i/A_0 \leq \frac{1}{2 \sigma \overline{\varrho}_f}$. 
The proof follows by invoking \Cref{cor:seed_height} to bound $H \geq \hat{h}_s\delta\lfs(p_i)$.
\end{proof}

\subsection{Octree refinement and outradii}\label{sec:octree}

Towards bounding the aspect ratio of Voronoi cells, we begin by proving some basic properties of the octree described in \Cref{sec:quality}. Given an octree box $\square_i$, denote by $c_i$ its center and $r_i$ its radius (half its diagonal length). Assume that the input $P$ has been scaled and shifted to fit into the unit cube $[0, 1]^3$. Starting with the unit cube as the box associated with the root node of the octree, the refinement process terminates with $r_i \leq \delta\lfs(c_i)$ for all leaf boxes $\square_i$. Note that refinement depends only on $\lfs$ and is independent of the number of points in $P$, and the distances between them. We establish the following Lipschitz-like properties for the size of leaf boxes.

\begin{apxClaim}\label{clm:boxrlfs}
If $\square_i$ is a leaf box, then $\frac{\delta}{2+\delta}\lfs(c_i) \leq r_i \leq \delta\lfs(c_i)$.
\end{apxClaim}
\begin{proof}
By definition the leaf box was not split, so $r_i \le \delta \lfs(c_i).$ Letting $\square_j$ be the parent of $\square_i$, it is clear that $\square_j$ had to be split. Hence, $r_j = 2r_i > \delta \lfs(c_j)$.
By Lipschitzness, $\lfs(c_i) \le \lfs(c_j) + r_i \le r_i ( 1 + 2 / \delta )$.
\end{proof}

\begin{apxClaim}\label{clm:boxrlfsp}
For any $p \in \square_i$, where $\square_i$ is a leaf box, $\frac{\delta}{2(1+\delta)} \leq r_i \leq \frac{\delta}{1 - \delta} \lfs(p)$.
\end{apxClaim}
\begin{proof}
Observe that $\dist(p, c_i) \leq r_i$, so $\lfs(p)$ is bounded in terms of $\lfs(c_i)$. Conveniently, \Cref{clm:boxrlfs} bounds $\lfs(c_i)$ in terms of $r_i$. To get the lower bound, we write $\lfs(p) \leq \lfs(c_i) + r_i \leq (\frac{2+\delta}{\delta}+1)r_i$. For the upper bound, we write $\lfs(p) \geq \lfs(c_i) - r_i \geq (1/\delta-1)r_i$.
\end{proof}

\begin{apxLemma}\label{lem:boxbalance}
If $\square_i$ and $\square_j$ are two leaf boxes sharing a corner, then $r_i/r_j \in [1/2, 2]$.
\end{apxLemma}
\begin{proof}
Assume that $r_j \leq r_i$.
From \Cref{clm:boxrlfs} we have $r_i \leq \delta\lfs(c_i)$ and $r_j \ge \frac{\delta}{2+\delta} \lfs(c_j)$. Together with $\lfs$ being 1-Lipschitz, this gives $r_j \ge \frac{\delta}{2+\delta}\big(\lfs(c_i) - (r_i+r_j)\big) \ge \frac{\delta}{2+\delta}(r_i/\delta - r_i- r_j).$
Simplifying, we get $r_j \ge \frac{r_i}{2} \frac{1-\delta}{1 + \delta}.$ For $\delta < 1/3$, we obtain $r_j > r_i /4$. As the ratio of box radii is a power of two, $r_j \in \{r_i/2, r_i\}$.
\end{proof}

These propoerties of the octree may be used to bound the outradius of Voronoi cells.

\begin{apxLemma}\label{lem:outradius}
The Voronoi cell of $s \in \seedset$ has outradius at most $\frac{2\delta}{1-3\delta} \lfs(s) \le \frac{4 (1+\delta)}{1-3\delta} r_i$, where $\square_i$ is the leaf box containing $s$.
\end{apxLemma}
\begin{proof}
Let $v$ be a vertex on the Voronoi cell of $s$. The octree construction guarantees $v \in \square_j$, for some leaf box $\square_j$. \Cref{clm:boxrlfsp} gives $r_j \le \delta/(1-\delta) \lfs(v).$
Fixing some $s' \in \square_j \cap \seedset \neq \emptyset$, it follows that $\dist(v, s) \leq \dist(v, s') \leq 2r_j$.
Hence, $\lfs(v) \geq \frac{1-\delta}{2\delta}\dist(v, s)$.
By Lipschitzness, $\lfs(s) \ge \lfs(v) - \dist(v,s) \ge \frac{1-3\delta}{2\delta} \dist(v,s).$
As $s \in \square_i$, \Cref{clm:boxrlfsp} gives $\lfs(s) \leq \frac{2(1+\delta)}{\delta}r_i$.
It follows that $\dist(v,s) \le \frac{2\delta}{1-3\delta} \lfs(s) \le \frac{4 (1+\delta)}{1-3\delta} r_i$.
\end{proof}

\subsection{Aspect ratio and size bounds}\label{sec:ar}
Any Voronoi vertex is in some box, and every box has at least one seed.
This provides an upper bound on the distance between a Voronoi vertex and its closest seed, and an upper bound on the cell outradius, for both interior and guide seeds. 
Interior seeds are at the center of a box containing no other seeds, so interior cell inradius is at least a constant factor times $r$. Combining the outradius and inradius bounds provides the following results.

\interiorAR*
\begin{proof}
Let $s \in \seedset$ be an interior seed and recall that $s$ was inserted at the center of some empty leaf box $\square_i$. By construction, $s$ is the only seed in $\square_i$. It follows that the inradius of $\Vor{s}$ is at least $\frac{1}{2\sqrt{3}}r_i$, which is half the distance from $c_i$ to any of its sides. 
The proof follows from the bound on the outradius in terms of $r_i$ as provided by \Cref{lem:outradius}.
\end{proof}

\boundaryAR*
\begin{proof}
Let $s \equiv s_{ijk} \in \seedset$ be a boundary seed and recall the lower bound of $\varrho_v\epsilon\lfs(p_i)$ on the inradius of $\Vor{s}$ from \Cref{lem:tet_inradius}.
By Lipschitzness, we may express this as $\varrho_v \delta (1-\delta) \lfs(s)$.
On the other hand, an upper bound of $\frac{4(1+\delta)}{1-3\delta}r_a$ on the circumradius of $\Vor{s}$ is provided by \Cref{lem:outradius}, where $\square_a$ is the leaf box containing $s$.
From \Cref{clm:boxrlfsp}, we have that $r_a \leq \frac{\delta}{1-\delta}\lfs(s)$.
With both bounds expressed in terms of $\lfs(s)$, we evaluate their ratio.
\end{proof}

\meshSize*
\begin{proof}
Let $\mathcal{I} = \seedset^\downarrow \cup \seedset^\ddarrow$ and $V(s)$ denote the Voronoi cell of seed $\as$.
Since the Voronoi cells of interior seeds in $\mathcal{I}$ partition the volume $\vol$, $\int_\vol \lfs^{-3} = \sum_{\aseed \in \mathcal{I}}  \int_{V(s)} \lfs^{-3}.$ Bounded outradii and inradii will bound each integral by as follows.

Fix a seed $\aseed$ and let $R_s$ and $r_s$ be the circumradius and inradius of $V(s)$, respectively. From \Cref{lem:outradius}, we have $R \leq \frac{2\delta}{1-3\delta} \lfs(\aseed)$. By Lipschitzness, for any $x \in \Vor{\aseed}$, $\lfs(x) \ge \frac{1-5\delta}{1-3\delta} \lfs(\aseed)$.
Thus, $\int_{\Vor{\aseed}} \lfs^{-3} \ge f_1(\delta) \lfs^{-3}(\aseed) \volfn{(\Vor{\aseed})}$, where $f_1(\delta) = \big( \frac{1-3\delta}{1-5\delta} \big)^3$.

If $s \in \seedset^\ddarrow$, \Cref{clm:boxrlfsp} yields $r_s \ge \frac{\delta}{4 \sqrt{3} (1+\delta)}\lfs(\aseed)$. Hence, $\volfn({\Vor{\aseed}}) \ge f_2(\delta) \lfs^3(\aseed)$, where 
$f_2(\delta) = \frac{4 \pi}{3} \left(\frac{\delta}{4 \sqrt{3} (1+\delta)}\right)^3$. If $s = s_{ijk} \in \seedset^\downarrow$, \Cref{lem:tet_inradius} gives 
$r_s \geq \varrho_v \epsilon \lfs(p_i)$. Recalling $\dist(p_i, s_{ijk}) = \delta \lfs(p_i)$ and the extension of $\lfs$ to the interior of $\vol$, we get $\lfs(\aseed) \le (1+\delta) \lfs(p_i)$. It follows that $r_s \geq \frac{\varrho_v \delta}{1+\delta} \lfs(\aseed)$ and $\volfn({\Vor{\aseed}}) \geq f_3(\delta) \lfs^3(\aseed)$, where $f_3(\delta) =  \frac{4 \pi}{3} \left( \frac{\varrho_v \delta}{1+\delta} \right)^3.$ 

Letting $f_4(\delta) = f_1(\delta)\cdot\min(f_2(\delta), f_3(\delta))$, we established that $\volfn({\Vor{\aseed}}) \ge f_4(\delta)\lfs^3(\aseed)$. Plugging that into the above bound, we get $\int_{\Vor{\aseed}} \lfs^{-3} \ge f_4(\delta)$.
Hence, $\int_\vol \lfs^{-3} \ge f_4(\delta)|\mathcal{I}| \ge f_4(\delta)|\seedset^\ddarrow|$.
The proof follows by observing that $\frac{1}{f_4(\delta)} \leq 18\sqrt{3}/\pi \cdot \epsilon^{-3}$.
\end{proof}

\end{document}